\documentclass[10pt]{article}
\usepackage{latexsym}
\usepackage{graphicx}
\usepackage{amsfonts}
\usepackage{amsmath, amsthm, amssymb}
\usepackage{fullpage}
\usepackage{setspace}
\usepackage{longtable}
\usepackage{natbib}
\bibpunct{(}{)}{;}{a}{}{,}
\usepackage{dcolumn}
\newcolumntype{.}{D{.}{.}{-1}}
\newcolumntype{d}[1]{D{.}{.}{#1}}
\usepackage{bm}
\usepackage{threeparttable}
\usepackage{booktabs}
\usepackage{enumerate}
\usepackage{bbm}
\usepackage{floatpag}
\usepackage{subfig}
\usepackage[font=small,labelfont=bf]{caption}

\usepackage{amsmath}
\usepackage{amsthm}
\usepackage{amssymb}
\usepackage{siunitx}
\usepackage{titlesec}
\usepackage[margin=1in]{geometry}

\titleformat{\section}
  {\normalfont\fontsize{10}{15}\bfseries}{\thesection}{1em}{}
  \titleformat{\subsection}
  {\normalfont\fontsize{10}{15}\bfseries}{\thesubsection}{1em}{}
  
  \titlespacing\section{0pt}{12pt plus 4pt minus 2pt}{0pt plus 2pt minus 2pt} 
  \titlespacing\subsection{1pt}{12pt plus 4pt minus 2pt}{0pt plus 2pt minus 2pt}

\newtheorem{lemma}{Lemma}
\newtheorem{lemma*}{Lemma}

\newtheorem{definition}{Definition}
\newtheorem{theorem}{Theorem}
\newtheorem{theorem*}{Theorem}

\newcommand{\bdk}{\mathbf{d}_k} 
\newcommand{\bY}{\mathbf{Y}}
\newcommand{\bX}{\mathbf{X}}
\newcommand{\bU}{\mathbf{U}}

\newcommand{\bZ}{\mathbf{Z}}

\title{Covariate-adaptive randomization inference in matched designs}
\author{Samuel D. Pimentel and Yaxuan Huang \thanks{Samuel D. Pimentel is an Assistant Professor and Yaxuan Huang is a doctoral student in the Statistics Department at University of California, Berkeley, 367 Evans Hall, Berkeley, CA 94720.  Correspondence should be addressed to spi@berkeley.edu. Samuel D. Pimentel is supported by the National Science Foundation under Grant No. 2142146, and both authors are supported by the Hellman Family Foundation.  {Part of this research was performed while Samuel D. Pimentel was visiting the Institute for Mathematical and Statistical Innovation (IMSI), which is supported by the National Science Foundation (Grant No. DMS-1929348).} The authors would like to thank Avi Feller, Kevin Guo, Ben Hansen, and Ruoqi Yu for helpful comments. }}

\begin{document}
\pagestyle{plain}

\maketitle

\begin{abstract}
It is common to conduct causal inference in matched observational studies by proceeding as though treatment assignments within matched sets are assigned uniformly at random and using this distribution as the basis for inference.  This approach ignores observed discrepancies in matched sets that may be consequential for the distribution of treatment, which are succinctly captured by within-set differences in the propensity score.  We address this problem via covariate-adaptive randomization inference, which modifies the permutation probabilities to vary with estimated propensity score discrepancies and avoids requirements to exclude matched pairs or model an outcome variable.  We show that the test achieves type I error control arbitrarily close to the nominal level when large samples are available for propensity score estimation. We characterize the large-sample behavior of the new randomization test  for a difference-in-means estimator of a constant additive effect.
 We also show that existing methods of sensitivity analysis generalize effectively to covariate-adaptive randomization inference.  Finally, we evaluate the empirical value of covariate-adaptive randomization procedures via comparisons to traditional uniform inference in matched designs with and without propensity score calipers and regression adjustment using simulations and analyses of genetic damage among welders and right-heart catheterization in surgical patients. { \textbf{Keywords:} causal inference, matching, permutation test, propensity score, sensitivity analysis.}
\end{abstract}

\onehalfspacing

\section{Introduction}

\subsection{Re-evaluating a common model for inference}

Randomized trials provide an effective means for measuring the effect of a treatment of interest relative to a control condition for at least two reasons.  First, random allocation of treatment to units ensures that large differences in pre-treatment characteristics between the group of units selected for treatment and the group of units selected for control arise in large samples only with very small probability. Second, the known distribution of indicators of treatment across units in the study provides a basis for inference.  By considering each unit's potential outcome values under treatment and under control as fixed latent values and repeatedly permuting treatment labels across study units (or using tests that rely on the large-sample behavior of such permutations) one may obtain inferences without making strong assumptions about the sampling procedure used to select the study units or the model for the outcome variable.  This conceptual approach, dating back at least  to \citet{fisher1935design}, is often known as randomization inference.

In contrast, in observational studies of a binary effect concerns arise about whether units receiving treatment and units receiving control are otherwise comparable.  If there is confounding, or systematic differences in variables (either observed or unobserved) that are predictive of the outcome of interest, the effect estimate from a simple group comparison will generally differ systematically from the effect that would have been measured in a randomized trial.  To adjust for observed confounding variables, researchers may estimate an outcome model in the absence of treatment and compare units with similar expected outcomes under control, estimate a treatment model or propensity score and compare units with similar propensities to receive treatment, or some combination of the two.  Matched observational studies adjust for observed confounding by grouping each treated unit with one or more similar control units and excluding controls not sufficiently similar to any treated unit \citep{stuart2010matching}.  When matching is conducted without replacement of controls so that matched sets are disjoint, and with exact agreement on a propensity score so that units grouped together shared identical propensities for treatment, then each matched set is like a miniature randomized trial; conditional on one unit within the group receiving treatment, each is uniformly likely to have been the one selected.  As such, methods of randomization inference are frequently applied  to matched observational studies as though in a stratified randomized trial.

Randomization inference in randomized trials depends on exact knowledge of the true randomization probabilities, and use of these methods in matched studies is motivated by an ideal setting in which the true propensity scores are known and matched exactly.  In reality, however, propensity scores must be estimated, and except in cases where the measured variables are few and discrete they are never matched exactly.  Optimal matching procedures such as those described in \citet{rosenbaum1989optimal}, \citet{zubizarreta2012using}, \citet{austin2015optimal}, and \citet{pimentel2020optimalmatching}  use estimated propensity score differences as important inputs, so that when a match is created the researcher has information easily available about which matched sets are relatively closer or further from achieving the ideal uniform distribution for treatment assignment.  Yet this information is not used when it comes time to do inference. By using uniform randomization inference, researchers implicitly assume a much simpler model and hope that these differences are all small enough not to create substantial lack of fit.  

We propose a new method for inference in matched observational studies, covariate-adaptive randomization inference, that explicitly uses estimated propensity score discrepancies to update permutation probabilities.  This approach retains most of the advantages of uniform randomization inference -- clear conceptual connections to a hypothetical randomized trial, ease of implementation, compatibility with interpretable methods of sensitivity analysis --- while addressing potential for lack of fit even in settings where propensity score differences need not disappear in large samples.   Furthermore, the lack of fit adjustment generally does not require users to alter the match itself in ways that reduce overall sample size. 
 
 In what follows we develop and evaluate covariate-adaptive randomization inference.  In Section \ref{sec:setup} we introduce a formal framework and describe the shortcomings of uniform randomization inference in detail.  In Section \ref{sec:adaptive} we introduce covariate-adaptive randomization inference, giving procedures for hypothesis testing and for estimation, and confidence interval construction under a constant additive effect model.
In Section \ref{sec:pscore_error} we bound the error introduced into the procedure by estimation of the propensity score.  In Section \ref{sec:sens}, we generalize sensitivity analysis procedures grounded in uniform randomization inference to covariate adaptive randomization inference to allow valid inference in the presence of unobserved confounding variables.  In Section \ref{sec:sims}, we demonstrate finite sample performance of covariate-adaptive hypothesis tests and confidence intervals and compare to alternative strategies such as caliper matching and regression adjustment.  In Section \ref{sec:data} we demonstrate implications for practice by reanalysis of two observational datasets: one measuring genetic damage experienced by welders and one assessing the impact of right-heart catheterization on patient mortality. Finally, Section \ref{sec:discussion} highlights important questions and connections raised by this work and outlines opportunities for further research.



\section{Formal framework and problem setup}
\label{sec:setup}

\subsection{Uniform randomization inference in matched designs}
\label{subsec:uniform}

Consider a population of individuals each represented by a vector $(Y(1), Y(0), Z, X, U)$. $Z$ is a binary indicator for membership in the treatment group, ${X}$ is a vector of observed covariates, and $U$ is an unobserved covariate.  The values $Y(z)$ are {potential outcomes} under different treatment conditions as defined under the Neyman-Rubin causal model and the stable unit variation assumption \citep{rubin1980randomization, holland1986statistics}, which specifies that an individual's outcome depends only on its own treatment status, rather than on the treatment status of other individuals, and that the only versions of treatment are 0 and 1. Only the information $(Y, Z, {X})$ is observed by the analyst, where $Y$ represents the observed outcome $Y(Z)$.  Let $\lambda(\mathbf{x}) = P(Z | {X} = {x})$ be the conditional probability of treatment given observed covariates, or the {propensity score}, and let $\pi({x}, u) = P(Z |{X} = {x}, U=u)$ be the true probability of treatment (which also depends on the unobserved $U$).

We assume individuals are first sampled independently from the population,  and then formed into a matched design $\mathcal{M}$, consisting of $K$ matched sets, on the basis of their treatment variables $Z$ and covariates $X$ alone.  Each matched set, numbered $k= 1$ through $K$, contains exactly one treated individual and one or more control individuals. Individuals in set $k$ are numbered $k1$ through $kn_k$, where $n_k$ is the number of units in set $k$, in arbitrary order, so that we may refer to the treatment indicator of the $i$th unit in the $k$th matched set as $Z_{ki}$.  We also define $n = \sum^K_{k=1}n_k$.   Let vectors $\bY(1), \bY(0), \bY, \bZ, \bU \in \mathbb{R}^n$ contain the observed univariate data ordered so that units in the same matched set are contiguous; in the language of  \citet{rosenbaum2002observational}[\S 2]. Let matrix $\mathbf{X}$ contain all $n$ vectors $X_{ki}$ in its rows with an identical ordering; in addition, we abuse notation slightly to  let $\lambda(\mathbf{X})$ represent the $n$-vector of true propensity scores. 
A matched design is said to be exact on a particular variable or quantity $\mathbf{v} \in \mathbb{R}^n$ if for any matched set $k$,  $v_{ki} = v_{kj}$ for all $i,j \in \{1, \ldots,n_k\}$.  
Let $\mathcal{Z}_\mathcal{M}$ be the set of all treatment vectors $\mathbf{Z}'$ such that $\sum^{n_k}_{i=1}Z'_{ki} = 1$ for all matched sets $k$.  

We now consider the the distribution of treatment assignments conditional on the match selected and the potential outcomes, i.e $P(\mathbf{Z} \mid \mathcal{F})$ where $\mathcal{F} = \{\mathcal{Z}_\mathcal{M}, \bX, \bY(1), \bY(0)\}$.  \citet[\S 3]{rosenbaum2002observational} argues that this distribution is discrete uniform on $\mathcal{Z}_\mathcal{M}$  under two key assumptions: first, the absence of unobserved confounding, under which $\lambda(x) = \pi(x,u) \text{ for all }u$, and exact matching on covariates $\mathbf{X}$ (or more generally exact matching on true propensity scores $\lambda(X)$).
 The null distribution for a test statistic $T$ can then be derived under the following sharp null hypothesis of no treatment effects for any individual in the sample:
\begin{align}
\label{eqn:sharpnull}
\mathbf{Y}(1) = \bY(0).
\end{align}
In particular, the sharp null hypothesis of no effect guarantees that the actual observed outcome $Y_i$ for individual $i$ would still have been observed had $Z_i$ taken on a different value. Thus under the sharp null the test statistic $T(Z,Y)$ varies conditional on $\mathcal{Z}_\mathcal{M}, \bX, \bY(1), \bY(0)$ only through the vector $\mathbf{Z}$, which is uniformly distributed over all permutations $\mathbf{Z}_{perm}$ of the observed element of $\mathbf{Z}$ within matched sets.  The exact p-value for the test of the sharp null, where we reject for larger values of $T(Z,Y)$, can be computed by computing the proportion of values of $\mathbf{Z}_{perm}$ for which $T(Z_{perm}, Y)$ exceeds $T(Z,Y)$.  In practice this quantity can be computed via Monte Carlo sampling from the permutation distribution or via a normal approximation to the permutation distribution.  For example, 
\citet{rosenbaum2002observational}[\S 2] 
gives asymptotic distributions for a variety of rank statistics.

Estimates of causal effects and confidence intervals can also be developed from the null distribution associated with a sharp null hypothesis.  Under a treatment effect model such as the constant additive model, the randomization test can be inverted to produce Hodges-Lehmann point estimates and corresponding confidence intervals.  Alternatively, when the test statistic itself is an estimator for an effect of interest (as in the case of the difference-in-means estimator), confidence intervals can be obtained under a normal approximation in large samples by using a variance estimate obtained from the null distribution and quantiles of the normal probability density function.

\subsection{Propensity score discrepancies imperil the uniform treatment distribution}


While the assumption of exact matching on a true propensity score yields convenient mathematical symmetry
 matching is never exact on all measured covariates in datasets of any substantial scale, nor is it practical to match exactly even on a univariate true propensity score, both because propensity scores must be estimated in practice and because they are often modeled as smooth functions of continuous variables that do not agree exactly for any two subjects.   As such it is important to consider the potential for resulting lack of fit between the nominal uniform distribution of treatment used for inference and the true distribution for a given matched design.  \citet{hansen2009propensity} addressed this question for the difference-in-means estimator of an additive treatment effect and found potential for slow-shrinking  finite sample bias when the true propensity score is not matched exactly.  
 Other authors have shown that bias in treatment effect estimation \citep{saevje2021inconsistency} and failure of Type I error control for the uniform randomization test \citep{guo2023statistical} persist even in infinite samples in settings where all treated units are matched in pairs without replacement  except in unusual situations such as populations where the probability of propensity scores exceeding 0.5 is zero or the true outcome model is known.  Of course, even bigger problems may arise if unobserved confounding is also present, but methods of sensitivity analysis have been constructed with specific attention to this issue, as will be explored in greater detail in Section \ref{subsec:sensreview}.
 
 It has been argued that certain kinds of balance tests may be understood as tests for lack of fit between the actual distribution of treatment and the uniform model of inference used \citep{hansen2009propensity}, so that if balance tests do not reveal problems then this problem can be ignored.  This approach falls short of resolving the problem optimally for two reasons.  First and most importantly, the theory underlying these balance tests relies on an asymptotic regime in which propensity score differences within matched sets approach zero as sample size increases, which in turn comes from a pattern of ever-larger concentrations of control units in a region arbitrarily close to any given treated unit. In many settings, this assumption is not reasonable.  For example, \citet{pimentel2015large} consider the common setting under which matches are conducted within natural blocks or groups of units, such as matching patients within hospitals or students within schools.  When, as in these examples, the size of an individual block may reasonably be viewed as bounded and the most natural way to think about increasing sample size is by adding more blocks, there is no reason to expect propensity score differences within matched sets to shrink to zero, since the concentration of matchable controls near a given treated unit is limited by the upper bound on the size of a block.  More generally, \citet{saevje2021inconsistency} demonstrated that when propensity scores larger than 0.5 are present in the population with probability exceeding 0 and matching is conducted without replacement, then some matched sets will necessarily have propensity score discrepancies bounded away from zero.  A second problem with the balance testing solution is that it does not fully articulate how to resolve problems with a matched design when the balance  test fails.  Common solutions such as using a tighter propensity score caliper lead to tradeoffs by reducing other aspects of match quality such as the proportion of treated units retained.

Another proposed strategy is the use of regression adjustment to remove slow-shrinking bias not addressed by close matching on a propensity score or on covariates \citet{abadie2011bias, guo2023statistical}.  Under assumptions on the outcome model, this method removes the bias, and under sufficiently strong assumptions on the outcome model and the convergence of matched discrepancies to zero, a fast rate is achieved.  However, the assumptions may not always be plausible, particularly the condition that the propensity score discrepancies shrink to zero as discussed by \citet{saevje2021inconsistency}.  In addition, estimating an outcome model may be inconvenient if the outcome is multivariate, is related to observed covariates in a complex or poorly-understood manner, or is not yet measured at the time the match is conducted.
%

\subsection{Lack of fit due to dependence between the true treatment vector and the match itself}
\label{subsec:zdependence}

All of the above discussion focuses on discrepancies  between the uniform distribution of treatment given $\mathcal{F} = \{\mathcal{Z}_\mathcal{M}, \bX, \bY(1), \bY(0)\}$ and the actual distribution of treatment given these quantities, thinking of the match $\mathcal{M}$ as fixed over all possible $\mathbf{Z}$-values.  However, if a model of independently-sampled subjects is assumed on the original data prior to matching then one might view the matched design, which is constructed with reference both to $\mathbf{X}$ and $\mathbf{Z}$, as a random function of treatment. In the special case of exact matching this issue need not arise, as the specific match chosen may remain conditionally independent of $\mathbf{Z}$ given event $\mathcal{Z}_\mathcal{M}$. However, when matching is not exact this guarantee no longer holds, and for some $\mathbf{Z} \in \mathcal{Z}_\mathcal{M}$ it may be the case that match $\mathcal{M}$ never would have been selected. 
 For example, a treated unit with a high propensity score may be difficult to match and be paired with a control having an appreciably lower propensity score; however, if the treated unit had been a control instead, it would have been possible to form a better match with a treated unit having a similarly high propensity score value. Values of $\mathbf{Z}$ that would have produced a different match do not belong in the support of the distribution of treatment conditional on $\mathcal{M}$, but they appear with positive support in the other two distributions mentioned. 
The phenomenon of reduced support will be denoted $Z$-dependence (with reference to the dependence of the match chosen on the original treatment status vector). 

 This issue is also discussed by \citet{pashley2021conditional}, who note, ``proper conditional analysis would need to take into account the matching algorithm."  However, no such analysis has yet been proposed in the matching literature, and \citeauthor{pashley2021conditional} suggest that matching may perform well empirically even in the presence of  $Z$-dependence.  $Z$-dependence is not a central focus in what follows, and unless otherwise noted we will take the perspective that the matched sets $\mathcal{M}$ are fixed.  However, $Z$-dependence will be helpful in explaining the empirical performance of uniform and covariate-adaptive randomization inference in certain settings, as will be shown, and we believe it is an important topic for further study.

\section{Adapting randomization inference to covariate discrepancies}
\label{sec:adaptive}

\subsection{True and estimated conditional distributions of treatment status}

To adapt randomization inference to discrepancies in observed covariates, we begin by considering the case where no unobserved confounding is present and represent the true conditional distribution in terms of propensity scores.  Since treatment indicator $Z_{ki}$ is Bernoulli$(\lambda(X_{ki}))$, we have: 
\begin{align*}
&P\left\{Z_{ki} = 1 \mid \mathcal{Z}_\mathcal{M}, \bX, \bY(1), \bY(0)\right\} \\
&=\frac{P\left\{Z_{ki} = 1, Z_{k2} = 0, \ldots, Z_{kn_k} = 0 \mid \lambda(\mathbf{X}_{k})\right\} }{
\sum_{j = 1}^{n_k}P\left\{Z_{k1} = 0, \ldots Z_{k(j-1)} = 0,Z_{kj} = 1,Z_{k(j+1)} = 0, \ldots, Z_{kn_k} = 0 \mid \lambda(\mathbf{X}_{k})\right\} } \\
&= \frac{\lambda(X_{ki})\prod_{j \neq i }[1 - \lambda(X_{kj})]}{\sum_{j=1}^{n_k}\lambda(X_{kj})\prod_{\ell \neq j }[1 - \lambda(X_{kj\ell)})]} = \frac{\text{odds}\{\lambda(X_{ki})\}}{\sum^{n_k}_{j=1}\text{odds}\{\lambda(X_{kj})\}}= p_{ki}.
\end{align*}
Because individuals are sampled independently, the joint conditional probability for a treatment vector $\mathbf{Z}$ is given by multiplying the appropriate $p_{ki}$ terms together:
\begin{align}
P\left\{\mathbf{Z}_\mathcal{M} = \mathbf{z} \mid \mathcal{Z}_\mathcal{M}, \bX, \bY(1), \bY(0)\right\} = \left\{ \begin{array}{cl} \prod_{k=1}^K\prod_{i=1}^{n_k}p_{ki}^{z_{ki}} & \mathbf{z} \in  \mathcal{Z}_\mathcal{M}\\ 0 & \mathbf{z} \notin \mathcal{Z}_\mathcal{M}\end{array} \right.
\label{eqn:true_distro}
\end{align}
If the true propensity score $\lambda(\cdot)$ were known, the $p_{ki}$s could be calculated exactly.  To conduct inference under the sharp null, the Monte Carlo strategy described in Section \ref{subsec:uniform} could be used, except that instead of permuting treatment assignments within matched sets uniformly at random one would sample the identity of the treated unit in each set from an independent multinomial random variable with 1 trial and sampling probabilities $p_{k1}, \ldots, p_{kn_k}$.  Exact finite-sample confidence intervals could also be obtained by inverting the test, and the large majority of the benefits of the uniform randomization inference procedure could be retained despite the reality of inexact matching on the propensity score.  Note that when propensity scores are matched exactly, this procedure reduces to the uniform test.

Unfortunately, the propensity score is typically unknown, so the true conditional distribution of treatment is also unknown.  However, it is straightforward to construct a plugin estimator for this distribution by substituting an estimate of the propensity score $\widehat{\lambda}(\cdot)$ for $\lambda(\cdot)$ in the formula for $p_{ki}$.  We denote the process of conducting hypothesis tests and creating confidence intervals using this plugin distribution as {covariate-adaptive randomization inference}.  In concrete terms, an investigator may conduct covariate-adaptive randomization inference via the Monte Carlo approach of Section \ref{subsec:uniform} by calculating the estimated propensity odds for each subject in a matched set, by dividing each estimated odds by the sum of the odds in the matched set,  by determining treatment via a single draw from a multinomial distribution with probabilities for each subject given by these normalized odds, and by conducting this process independently within all matched sets, repeating as necessary to generate a close approximation to the null distribution.  


\subsection{The difference-in-means estimator: large-sample distribution}
\label{subsec:largesample}

While Monte Carlo sampling from the covariate-adaptive permutation distribution is straightforward, it is also instructive to construct a normal approximation using the mean and variance of the test statistic.  
Here we demonstrate this approach for the standard difference-in-means estimator:
\[
T(\mathbf{Z}_\mathcal{M},\mathbf{Y}_\mathcal{M})= \frac{1}{K}\sum^K_{k=1}\left\{\sum^{n_k}_{i=1}Z_{ki}Y_{ki} - \frac{1}{n_k-1}\sum^{n_k}_{i=1}(1 - Z_{ki})Y_{ki}\right\} = \frac{1}{K}\sum^K_{k=1}\sum^{n_k}_{i=1}Y_{ki}\frac{n_kZ_{ki} - 1}{n_k -1}. 
\]
Under the sharp null hypothesis, the $\mathbf{Y}$ is fixed across all values for $\mathbf{Z}$ so the expectation and the variance are calculated only over the random variables $\mathbf{Z}$.
\begin{align}
\label{eqn:nullmean}
E(T&(\mathbf{Z}_\mathcal{M},\mathbf{Y}_\mathcal{M} )\mid\mathcal{Z}_\mathcal{M}, \bX_\mathcal{M}, \bY_\mathcal{M}(1), \bY_\mathcal{M}(0)) 
=\frac{1}{K}\sum^K_{k=1}\sum^{n_k}_{i=1}Y_{ki}\frac{n_kp_{ki} - 1}{n_k - 1}\\\nonumber
Var(&T(\mathbf{Z}_\mathcal{M},\mathbf{Y}_\mathcal{M}) \mid\mathcal{Z}_\mathcal{M}, \bX_\mathcal{M}, \bY_\mathcal{M}(1), \bY_\mathcal{M}(0)) 
= \frac{1}{K^2}\sum^K_{k=1}\sum^{n_k}_{i=1}\left(\frac{n_k}{n_k-1}\right)^2Y_{ki}p_{ki}\left\{Y_{ki}(1- p_{ki}) - \sum^{n_k}_{j\neq i}Y_{kj}p_{kj}\right\} 
\end{align}

Note that these formulas get considerably simpler in the case of a matched pair design, in which $n_k = 2$ for all $k$.  Here the inner sum  $\sum^{n_k}_{i=1}Y_{ki}\frac{n_kZ_{ki} - 1}{n_k -1}$ may be rewritten as $(Z_{k1} - Z_{k2})(Y_{k1} - Y_{k2})$, so that the mean of the test statistic is the sample average (across matched pairs $k$) of $V_kD_k$ where $V_k = p_{k1} - p_{k2}$ and $D_k = Y_{k1} - Y_{k2}$, and the variance is the sample average of $K^{-1}(1 - V_k^2)D_k^2$.

While this mean and variance  of the test statistic depend on the unknown true propensity score through the $p_{ki}$ terms, they can be estimated by substituting the $\widehat{p}_{ki}$ obtained by substituting the estimated propensity score $\widehat{\lambda}$ for the true propensity score $\lambda$. Large-sample inference may be conducted by computing a normalized deviate of the difference-in-mean statistic using estimates of the mean and variance above, and comparing to the standard normal distribution. This relies on a finite-sample central limit theorem based on a sequence of infinitely-expanding finite samples \citep{li2017general}  for which the mean and variance of the test statistic converge to stable limits.  
The most natural asymptotic regime places some upper bound on the size of a matched set $n_k$ across all these samples while $K$ grows towards infinity, since returns to matching large numbers of treated units to a single controls are quickly diminishing \citep{hansen2004full}; for more discussion of asymptotic regimes in matching with restricted or fixed matching ratios, see \citet{abadie2006large}.  The simplest available  central theorem is thus in the inner sums $\sum^{n_k}_{i=1}Y_{ki}\frac{n_kZ_{ki} - 1}{n_k -1}$, which are independent random variables with non-identical distributions, and can be proved using standard approaches when $n_k$ is uniformly bounded and a Lindeberg condition holds on the potential outcomes $Y(0)$.

The large-sample approximation just described does not explicitly account for random variation due to the estimation of the propensity score; in particular, the mean term, which is a function of  $\widehat{p}_{ki}$ random variables, is instead treated as fixed.  To probe the possible role of such variation, we used the M-estimation framework  \citep{stefanski2002calculus} to represent both the propensity score estimates and the test statistic of interest as part of a single multivariate estimation problem and to construct a new variance estimate for the test statistic that attempts to incorporate variation in the estimated propensity scores.  However, the new variance estimates were frequently near-identical or slightly smaller than those ignoring the propensity score estimation (presumably due to positive correlation between the estimated mean and the test statistic across propensity score fits); in our simulations the associated tests performed almost identically.  Given the simpler form and intuitive permutation analogue for the test ignoring propensity score estimation, we focus on this test going forward.  More details are provided on the M-estimation approach in Section A.3 of the appendix.


 \subsection{Estimates and confidence intervals for an additive effect}

So far we have discussed covariate-adaptive randomization inference primarily through the lens of testing the sharp null hypothesis of zero effect.  
It is also possible to use covariate-adaptive randomization inference to construct confidence intervals in setting where the primary goal is estimating a specific causal effect.  The exact details of this process depend on the estimand; for purposes of exposition, we focus on estimating a constant additive treatment effect.  More formally, we assume that there exists $\tau$ such that
\[
Y_i(1) = Y_i(0) + \tau \quad \quad \text{for all $i$}.
\]
and we seek to estimate $\tau$. 

One way to leverage a randomization inference procedure to produce such estimates is Hodges-Lehmann estimation.   First, note that for $\tau > 0$, the sharp null hypothesis of no effect no longer holds and observed outcomes $Y_i$ are no longer invariant to treatment status, but  the transformed outcomes $Y_i - Z_i\tau$ will be invariant, and a randomization test can be constructed by using $Y_i - Z_i\tau$ as input to the test statistic rather than $Y_i$.  Hodges-Lehmann estimation considers the resulting transformed test statistic and asks which value of $\tau$ makes it closest in value to its expectation under the null hypothesis.  For a helpful review in the context of randomization inference in observational studies, see \citet[\S 2.7.2]{rosenbaum2002observational}.  Hodges-Lehmann estimates for treatment effect parameters are consistent under mild regularity assumptions \citep[\S 1.4]{maritz1995distribution}

To conduct Hodges-Lehmann estimation of $\tau$, 
let $D(\tau)$ represent the difference in means of the transformed outcome $Y_i - Z_i\tau$:
\begin{align*}
D(\tau) &= \frac{1}{K}\sum^K_{k=1}\sum^{n_k}_{i=1}(Y_{ki} - \tau Z_{ki})\frac{n_kZ_{ki} - 1}{n_k -1}\\
&= \frac{1}{K}\sum^K_{k=1}\sum^{n_k}_{i=1}Y_{ki}\frac{n_kZ_{ki} - 1}{n_k -1} - \frac{1}{K}\sum^K_{k=1} \tau\frac{n_k - 1}{n_k -1}\\
&= D(0) - \tau.
\end{align*}
The Hodges-Lehmann estimate $\widehat{\tau}_{HL}$ is the value of $\tau$ such that $D(\tau) =D(0)-\tau$ is equal to  t$E[D(0)]$, so 
\[
\widehat{\tau}_{HL}  = D(0) - E[D(0)]
\]

In practice, we conduct covariate adaptive randomization inference using the estimated propensity score $\widehat{\lambda}(\cdot)$ and the associated estimates $\widehat{p}_{ki}$.  The discrepancy between the Hodges-Lehmann estimate for this distribution and the corresponding Hodges-Lehmann estimate for true distribution is now given as follows:
\begin{align}
\label{eqn:finsamp_weighted}
\frac{1}{K}\sum^K_{k=1}\sum^{n_k}_{i=1}Y_{ki}\frac{n_k(\widehat{p}_{ki} - p_{ki}) - 1}{n_k - 1}
\end{align}
To the degree the $\widehat{v}_k$s are better estimates of the true $p_{ki}$s than the uniform estimate of $1/n_k$, model-weighted randomization inference recovers a more accurate approximation of the target Hodges-Lehmann estimator than the uniform randomization distribution.  

To obtain confidence intervals for the parameter $\tau$, one may either invert the Monte Carlo randomization test by searching over values $\tau_0$ for which the sharp null hypothesis $\tau = \tau_0$ is not rejected using test statistic $D(\tau_0)$, or one may use a normal approximation for the Hodges-Lehmann estimate.  Since the Hodges-Lehmann estimate differs from the usual difference-in-means only by a constant shift (given the estimated propensity score), their variances are identical and the formula in Section \ref{subsec:largesample} can be used to determine the width of the confidence intervals.

\section{Assessing the impact of propensity score estimation error}
\label{sec:pscore_error}

A primary concern in determining the empirical value of covariate-adaptive randomization inference is understanding the impact of errors in estimating the propensity score. To trust the method, we need some guarantee that small deviations between $\widehat{\lambda}(x)$ and true $\lambda(x)$ values result in only small deviations in nominal and actual test size and confidence interval coverage.  The following results provide partial reassurance in this regard.  To lay the groundwork, let $\widehat{\lambda}_N(\cdot)$ be an estimated propensity score fit on an external sample of size $N$ from the infinite population, and let the quantities $\mathbf{Z},\mathbf{X}, \mathbf{Y}(1), \mathbf{Y}(0)$ refer to a separate sample organized into $K$ matched pairs containing $\sum_{k=1}^Kn_k = n$ total units.   Let $F_\lambda$ represent the  true conditional distribution of $\mathbf{Z}$ as a function of the true propensity score $\lambda(\cdot)$ and let $F_{\widehat{\lambda}_N}$ represent the distribution of $\mathbf{Z}$ used to conduct covariate-adaptive inference in practice, based on $\widehat{\lambda}_N(\cdot)$.  Furthermore,  let $p_{\widehat{\lambda}_N,n}$
be the p-value produced by a nominal level-$\alpha$ covariate-adaptive randomization test of the sharp null hypothesis using 
$F_{\widehat{\lambda}_N}$.    The following result, adapted from \citet{berrett2020conditional} (who consider the slightly different case in which permutations are across the entire dataset and not within matched sets), relates these quantities using the total variation distance of $F_\lambda$ and $F_{\widehat{\lambda}_N}$.
\begin{theorem}
If no unobserved confounding is present and the sharp null hypothesis of no effect is true, then
\[
P(p_{\widehat{\lambda}_N,n} \leq \alpha) \,\, \leq \,\, \alpha + d_{TV}\left(F_\lambda, F_{\widehat{\lambda}}\right)
\]
where $d_{TV}(P,Q)$ gives the total variation distance between two probability distributions $P$ and $Q$.
\label{thm:berrett}
\end{theorem}

In words, this theorem says that the true type I error of a covariate-adaptive randomization test performed with nominal level $\alpha$ is no larger than $\alpha$ plus the total-variation discrepancy of the distributions implied by the true and estimated propensity score.  Intuitively, when the estimated propensity score is close to the true propensity score, this means that the nominal type I error rate is close to the true type I error rate.  We formalize this intuition for the case in which the true propensity score obeys a logistic regression model.

\begin{theorem}
Suppose that $P(Z=1  \mid X) = \lambda(X) = \frac{1}{1 + \exp\left(-\beta^TX \right)}$ and that $\widehat{\lambda}_N$ is obtained by estimating this model using maximum likelihood. Suppose furthermore that $N$ increases with the primary sample size $n$ such that $\lim_{n\longrightarrow \infty} n/N = 0$, and suppose the covariates $X$ have compact support.  Then under the conditions of Theorem \ref{thm:berrett},
\[
\limsup_{n,N \longrightarrow \infty}P(p_{\widehat{\lambda}_N,n} \leq \alpha) \,\, \leq \,\,  \alpha.
\]
\label{thm:logistic}
 \end{theorem}
The proof, which uses Pinsker's inequality to bound the total variation distance by a sum of Kullback-Leibler divergences and a Taylor expansion to show that this sum converges to zero, is deferred to Section A.1 of the appendix.

A natural question is whether similar results can be obtained when the propensity model is not necessarily logistic. \citet{berrett2020conditional} sketch a proof for a result similar to Theorem \ref{thm:logistic} when the propensity score is estimated nonparametrically using a kernel method; this requires only mild smoothness conditions rather than a correctly-specified model, although the bound on the rate of convergence of the size of the Type I error violation towards zero is much weaker.  More generally, one may turn to Theorem \ref{thm:berrett} directly to explore misspecification: this result provides a bound on Type I error violation whenever the degree of misspecification can be quantified by the total variation distance between estimated and true conditional distributions of treatment.  Such metrics for probing robustness appear elsewhere in the literature; for example, \citet{guo2022partial} assess robustness of inferences to covariate noise by considering a regime in which the true conditional distribution of covariates given treatment has bounded total variation distance from the true conditional distribution.

Another limitation of Theorem \ref{thm:logistic}, which extends also to the nonparametric kernel approach just mentioned, is the asymptotic regime.  The issues extend beyond the usual question of whether a particular sample size is large enough to ensure that error is small; here the requirement that the pilot sample used to estimate the propensity score grow at a larger rate than the analysis sample raises similar concerns even for very large samples, since the key question is whether the pilot sample's size sufficiently exceeds that of the analysis sample. Fitting the propensity score in a large independent sample consisting of most of the observed data while reserving a relatively smaller portion for the analysis is uncommon in applied matching studies. However, we note that this approach is natural in settings where treatment and covariate values are readily available but outcomes are expensive or difficult to measure, as when researchers must collect outcomes by administering a test to study subjects \citep{reinisch1995utero} or by abstracting medical charts \citep{silber2001multivariate}.  In Section \ref{sec:sims}, we evaluate the performance of the method by simulation in multiple settings, including some that plausibly adhere to the assumed asymptotic regime and others that likely do not.  In general we find only minor difference in performance for tests based on estimated propensity scores fit out-of-sample in samples much larger than the analysis samples versus tests based on propensity scores estimated in-sample on the analysis data, suggesting that the test may be fairly robust to violations of the asymptotic regime given in Theorem \ref{thm:logistic} in many finite sample settings.

 \section{Covariate-adaptive randomization inference under unobserved confounding}
\label{sec:sens}

 \subsection{Review of sensitivity analysis framework under exact matching}
 \label{subsec:sensreview}

Results in Sections \ref{sec:setup} -\ref{sec:pscore_error} have all depended on the absence of unobserved confounding, but in practice it is not plausible that all confounders are observed.   To address this issue with the matched randomization inference framework, \citet[\S 4]{rosenbaum2002observational}  presents a method of {sensitivity analysis} to relax the no unobserved confounding 
assumption.  Specifically, the true probabilities of treatment are restricted as follows:
\begin{align}
\label{eqn:gammabound}
1/\Gamma \leq \frac{\pi(x,u)( 1-\pi(x,u')}{\pi(x,u')(1 - \pi(x,u))} \leq \Gamma \quad \text{for all $x,u,u'$}.
\end{align}
This is equivalent to the following model for treatment assignment, with arbitrary $\kappa(\cdot)$ and  and $\lambda (X_i) = \text{logit}^{-1}(\kappa({X}_i))$:
\begin{align}
\label{eqn:logitmodel}
\log\left(\frac{\pi(X_i, U_i)}{1 - \pi(X_i, U_i)} \right) = \kappa({X_i}) + \gamma U_i \quad \text{where }\gamma = \log(\Gamma)\text{ and } 0 \leq U_i \leq 1.
\end{align}
To complete the sensitivity analysis, some method is needed to compute worst-case p-values over all possible values of of $\mathbf{U}$ allowed by the model.  
The method depends on the structure of matched sets formed and the test statistic used; we focus on the approach of \citet{rosenbaum2018sensitivity} which allows for arbitrary strata structure and applies to any sum statistic, i.e. any statistic $T(\mathbf{Z},\mathbf{Y}) = \sum^K_{k=1}\sum^{n_k}_{i=1}Z_{ki}f_{ki}(\mathbf{Y})$  
for chosen functions  $f_{ki}$ of $\mathbf{Y}$. As shown in \citet[\S 5]{rosenbaum2018sensitivity}, the difference-in-means statistic, the bias-corrected difference-in-means statistic of the previous section, and other M-statistics are all members of the family of sum statistics.

Without loss of generality, suppose the $n_k$ units in each matched set $k$ are arranged in increasing order of  their values $f_{ki}$ so $f_{k1} \leq f_{k2} \leq \ldots, f_{kn_k}$ for all $k$, and suppose that we are interested in a one-sided test where larger values of the test statistic will lead to rejection.  Let $U^+$ be the set of all stratified N-tuples $\mathbf{u}$ such that $u_{ki} \in \{0,1\}$ and $u_{k1} \leq u_{k2} \leq \ldots \leq u_{kn_k}$ for all $k$, and let $U^-$ be the set of all stratified N-tuples $\mathbf{u}$ such that $u_{ki} \in \{0,1\}$ and $u_{k1} \geq u_{k2} \geq \ldots \geq u_{kn_k}$ for all $k$.  Let $\alpha_{unif}(\mathbf{Y},\mathbf{u})$ represent the p-value for the uniform randomization test performed with test statistic $T(\mathbf{Z}, \mathbf{Y})$ and the uniform randomization distribution when model (\ref{eqn:logitmodel}) holds with unobserved confounder $\mathbf{U} = \mathbf{u}$.  When matching on the propensity score is exact, \citet{rosenbaum1990sensitivity} showed the following:
\[
\min_{\mathbf{u} \in U^-}\alpha_{unif}(\mathbf{Y},\mathbf{u}) \leq \alpha_{unif}(\mathbf{Y},\mathbf{U}) \leq \max_{\mathbf{u} \in U^+}\alpha_{unif}(\mathbf{Y},\mathbf{u}).
\]
Although $\mathbf{U}$ is still unknown, this statement allows us to bound its impact on the result of the hypothesis test by searching over a highly structured finite set of candidate $\mathbf{u}$-values.
We now show that this approach still works to identify the worst-case p-value when matches are not exact on the propensity score and permutation probabilities are covariate-adaptive.

\subsection{Sensitivity analysis model under covariate-adaptive randomization inference.}

As in the previous section, let  $T(\mathbf{Z}_\mathcal{M}, \mathbf{Y}_\mathcal{M})$ be a sum statistic and focus on the case of a one-sided where large values lead to rejection.  We now define $\alpha_{adapt}(\mathbf{Y},\mathbf{u})$ as the p-value obtained by conducting a covariate-adaptive randomization test (using the true propensity score) when model (\ref{eqn:logitmodel}) holds with unobserved confounder $\mathbf{U} = \mathbf{u}$.

\begin{theorem}
\label{prop:omega_bound}
For any $\mathbf{u} \in [0,1]^n$,
$\min_{\mathbf{u}' \in U^-}\alpha_{adapt}(\mathbf{Y},\mathbf{u}') \leq \alpha_{adapt}(\mathbf{Y},\mathbf{u}) \leq \max_{\mathbf{u}'' \in U^+}\alpha_{adapt}(\mathbf{Y},\mathbf{u}'').
$
\end{theorem}
The proof, which is adapted from results for the exact-matching case in \citet{rosenbaum2002observational}, is deferred to Section A.2 of the appendix; the argument is not specific to the p-value and may also be applied to obtain a broader class of functions of the covariate-adaptive randomization distribution, including Hodges-Lehmann estimates. While the result is proved under the assumption that true propensity scores are used, the bounds here may be combined with the bounds of Theorem \ref{thm:berrett} to obtain guarantees with an estimated propensity score.

The bound in Theorem \ref{prop:omega_bound} gives us a finite set of possible distributions over which to search to identify the worst-case p-value for a covariate-adaptive randomization test.  However, under certain configurations of strata size and number, this set may grow large and complicated so that it is difficult to compute the exact maximum efficiently.  
\citet{rosenbaum2018sensitivity} provided a close, conservative approximation to the worst-case p-value that aggregates over easily-computed quantities each using information from only a single stratum.  Specifically, the calculations for the upper bound depend on the mean $\mu_{k\ell}$ and variance $\nu_{k\ell}$ of the stratum-specific contribution to the test statistic given by $\sum^{n_k}_{i=1}Z_{ki}f_{ki}(\mathbf{Y}_\mathcal{M})$ when $u_{k1} = \ldots =  u_{k\ell} = 0 $ and $u_{k(\ell + 1)} =  \ldots, u_{kn_k} = 1$ for each stratum $k$ and each value $\ell = 0,1, \ldots, n_k$, given the value for the sensitivity parameter $\Gamma$ (for the lower bound similar quantities are used but with a different choice of binary $u_{ki}$ values).

 Under covariate-adaptive randomization inference these quantities take on the following values:
\begin{align*}
\mu_{k\ell} = \frac{\sum^{\ell}_{i =1}p_{ki}f_{ki} + \Gamma\sum^{n_k}_{\ell + 1}p_{ki}f_{ki}}{\sum^{\ell}_{i =1}p_{ki} + \Gamma\sum^{n_k}_{\ell + 1}p_{ki}} \quad \quad \quad 
\nu_{k\ell} = \frac{\sum^{\ell}_{i =1}p_{ki}f_{ki}^2 + \Gamma\sum^{n_k}_{\ell + 1}p_{ki}f_{ki}^2}{\sum^{\ell}_{i =1}p_{ki} + \Gamma\sum^{n_k}_{\ell + 1}p_{ki}} - \mu_{k\ell}^2
\end{align*}
The method for combining these quantities to provide an overall bound on the p-value is identical to that described in \citet{rosenbaum2018sensitivity}.

\section{Finite-sample performance via simulation}
\label{sec:sims}

%
%

\subsection{Simulation setup}
The results of Section \ref{sec:pscore_error} provide some assurance that in extremely large samples, covariate-adaptive randomization inference will closely approximate the true conditional distribution of treatment and outperform uniform randomization inference.  However, many matched studies work with relatively small sample sizes, and researchers do not always have strong guarantees that fitted models are well-specified.  In what follows we assess the empirical performance of covariate-adaptive randomization inference in small to moderate samples via a simulation, considering cases where models are correctly and incorrectly specified.

We create a matrix of covariates $X$ by drawing $p$ vectors of independent 
standard normal random variables, each vector of length $n$.  $p$ is 
2, 5, or 10, and $n$ is either 100 or 1000. The true propensity score is then given by one of the following 
two 
functions, one that specifies the logit of treatment as a linear function of the 
columns of $X$ and
another specifying it as a nonlinear function of those columns:
\begin{align}
\text{logit}\left[P(Z=1 \mid X) \right] &= \log\left(\frac{0.3}{0.7}\right) + 
    \tau \cdot X_1\\
\text{logit}\left[P(Z=1 \mid X) \right] &= \log\left(\frac{0.3}{0.7}\right) + 
    \frac{\tau}{\sqrt{265}}\left(X_1 + 4X_1^3 \right)
\end{align}
The functional forms of these models are adapted from those used by Resa and
Zubizarreta (\emph{Stat. in Med.} 2016).  The parameter $\tau$ controls the strength of the propensity score signal; we
examine two signals, ``weak signal" with $\tau = 0.2$ and ``strong signal" with
$\tau = 0.6$.  The intercept is chosen to ensure a a treated:control ratio in the
general neighborhood of 3:7 (since all  $X_i$ variables have mean zero over 
repeated samples), 
providing a relatively large control pool for pair
matching. 
The scaling factor $1/\sqrt{265}$ 
chosen to render the overall signal-to-noise ratio 
comparable across 
models for a given value of $\tau$.  
Once the true propensity  score has been established, indicators for treatment
are drawn according to the model. 
For a more detailed look at the distribution of the propensity scores in the 
treated and control groups under this setup, see Figure 4 
in Section A.4 of the appendix.

Matching is conducted on the joint $(X,Z)$ datasets created by this process.
The matching is done using a robust Mahalanobis distance on the columns of $X$,
either with or without a propensity score caliper.  When the caliper is used,
a propensity score must first be estimated.  This is done by fitting a logistic
model with linear additive terms using maximum likelihood. 
Then matches are restricted
to occur only between individuals separated by no more than 0.2 sample  standard
deviations  of the fitted propensity scores (computed for the original dataset);
if treated units must be excluded from the match to meet this condition  the 
minimal possible number of such exclusions is made.

Finally, outcomes are drawn using a linear model with the same right-hand side
as the model used to generate the true propensity score, albeit with no intercept, with $\tau = 1$, and with additive independent
mean-zero normal errors with variance 4. 
We conduct randomization inference
by reshuffling treatment indicators within pairs uniformly at random,
based on propensity scores estimated as described above, or based on true 
propensity scores.
  The
observed test statistic is compared to the null distribution obtained by 
Monte Carlo sampling of 5000 draws to see  whether one-sided tests for a difference greater than zero
 reject
at the 0.05
level.  Results are evaluated for both raw outcomes  and outcomes adjusted by ordinary least-squares linear regression using the approach of \citet{rosenbaum2002covariance}.

In addition, for each combination of simulation parameters we also test the results of inference when the observed
test statistic is not generated by the original $\mathbf{Z}$-vector but by a within-match permutation
of treatment indicators using distribution (\ref{eqn:true_distro}) with the true propensity scores.  This is designed to
approximate a setting in which matched pairs are sampled independently, rather than the typical setting in which individual subjects are sampled independently prior to matching,
and eliminates the potential for $Z$-dependence as discussed in Section \ref{subsec:zdependence}.


For datasets of size $n=100, p = 2$ and $n=100, p = 5$ we ran each unique combination of the remaining simulation
parameters (propensity model, signal strength, 
caliper indicator, 
 inference method, regression adjustment
indicator, and indicator for reshuffled $Z$-vector in observed test statistic)  8000 times.  Type I error rates are calculated as the sample average
of rejection indicators over the 8000 draws for each parameter combination.  We also ran each combination of parameters 8000 times for datasets of size $n=1000, p=10$.  We opted not to examine the $n=100, p=10$ case because we wanted to focus on cases in which completed matches met standard balance criteria and sample runs suggested that most matches were not successful in balancing all ten covariates, and we studied only $p=10$ for the much more computationally expensive $n=1000$ case in order to produce simulation results in a reasonable timeframe.

\subsection{Results}
\label{subsec:simresults}

Figure \ref{fig:sim_heatmaps} shows the primary results on Type I error rate from the simulations.  The most apparent pattern is the high type I error rates for uniform inference on uncalipered matches with the difference-in-means statistic, much higher than the rates for any method that accounts in any additional way for in-pair covariate discrepancies.  This is true across every simulation setting shown but especially so at $n=1000$.  

More nuanced are the distinctions among settings using covariate-adaptive inference, calipers, regression-adjusted test statistics, and combinations thereof.  In the first three columns of each table, corresponding to settings in which the sampling model is on the rows of the original dataframe and the match itself is viewed as a random variable, settings with calipers or regression adjusted test statistics tend to perform best, either with or without covariate-adaptive inference; in contrast settings relying entirely on covariate-adaptive inference tend to do slightly worse or even substantially worse when $n=1000$.  However, when sampling is conducted over random treatment assignments within matched pairs (a setting in which $Z$-dependence is eliminated), covariate-adaptive inference alone is generally just as effective as regression adjustment and does  better than calipers.  

Another lens through which to view the results is across the four horizontal groupings in each table, each of which denotes a different combination of correctly/incorrectly specified treatment and outcome models.  Overall covariate-adaptive inference with estimated propensity scores tends to fare more poorly when the treatment model is misspecified, even when $Z$-dependence is not present.  While regression adjustment is somewhat robust to the degree of nonlinearity introduced in this simulation's outcome model, when both outcome and treatment models are nonlinear type I error control often fails for $p > 2$.  At $n=1000$ the case with both models misspecified leads to gross violation of type I control under any form of adjustment except oracle covariate-adaptive inference (which uses true, nonlinear treatment probabilities that are unavailable in practice) although using some form of adjustment is still a vast improvement on raw uniform inference.

All the results shown above focus on the strong propensity signal case ($\tau = 0.6$), in which matching tends to be substantially more difficult than in the weak signal case.  Results for the  Density plots and type I error rates 
for the weak-signal case are given in Figures 5 and 6 in Section A.4 of the appendix.  While the overall ordering among approaches remains similar, with uniform inference absent any other adjustments underperforming methods with some form of covariate adjustment and covariate-adaptive inference alone lagging other methods slightly when $Z$-dependence is present, the size of type I error violations is greatly reduced 
such that the simple uniform approach often controls Type I error when $p=2$.  

We also conducted two robustness checks to confirm that our results are not sensitive to secondary aspects of our simulation setting. 
Firstly, 
 to address concerns about whether all the automated matches we performed were successful in balancing covariates, we reproduced the Type I error results excluding all simulation runs in which a matched sample failed to achieve absolute standardized differences less than 0.2 on all measured covariates. Secondly, 
in recognition  that our theory results rely on propensity scores estimated in pilot samples larger in size than the actual analysis sample, we repeated our simulations with propensity scores fit on independent pilot samples of size 10,000, simulated from the same data-generating process as the analysis sample.  In general, these robustness checks all substantiated the patterns observed in the primary simulations.  For full results, see Figures 7-8 in Section A.4 of the appendix.
 
Finally, in order to further elucidate differences among adjustment strategies that appear comparable in Type I error control, we study precision of inference, as measured by confidence interval width (Figure \ref{fig:sim_heatmaps_ci}).  Specifically, for all simulation settings that achieved approximate Type I control, in the sense that the Bonferroni-corrected test against the null that the true error rate is $0.05$ did not reject with a positive $z$-statistic, we used the central limit theorem approximation of Section \ref{subsec:largesample} to create large-sample confidence intervals based on the normal distribution for an additive shift effect.  For cases with regression-adjusted test statistics, the procedures of Section \ref{subsec:largesample} were used with residuals from the regression model in place of raw outcomes $Y_i$.  The two main takeaways are the value of regression adjustment for improving precision, a pattern that appears in other parts of the causal inference literature \citep{rosenbaum2002covariance, fogarty2018regression, antonelli2018doubly}, and a cost in precision associated with imposing calipers.  This cost comes from reduced sample size, since treated units may be excluded by the caliper leading to fewer matched sets formed.

In summary, the simulations reveal several important messages.  First, they emphasize the importance of taking measures of some kind beyond optimal matching itself to control covariate discrepancies within matched pairs and guarantee Type I error control for post-match inference, at least in settings with any marked degree of distinction between propensity score distributions in the treated and control samples.  Moreover, the importance of such measures appears to increase with the size and complexity of the dataset.  With respect to Type I error control, the type of correction or adjustment employed seems to matter less than the fact of employing it.   Secondly,  while covariate-adaptive inference in isolation is slightly less successful than calipers or regression adjustment in many settings, this appears to be largely a product of $Z$-dependence, since changing the sampling approach to view pairs as the units of analysis rather than pre-match study units tends to erase this gap in most cases.  This finding suggests the importance of the specific stochastic model used to justify matched randomization inference, and the value of further work on understanding $Z$-dependence and finding effective ways to remove it.  Finally, the simulations show that among adjustment strategies that do not rely on fitting some form of outcome model, covariate-adaptive inference holds a small edge over caliper approaches because it does not require a reduction in sample size to fix problems with in-pair covariate discrepancies.


  \section{Performance in case studies}
  \label{sec:data}

\subsection{Welders and genetic damage}

We now  apply the tools developed for covariate-adaptive randomization
inference to a real dataset due originally to \citet{costa1993dna}.  This dataset compares the rate
 of DNA-protein  cross-links, a risk factor  for gene expression problems,  among
 welders  and controls in other occupations.  In addition to DNA-protein linkage,
 age, race, and smoking behavior  are measured for  all subjects (each of whom
 is male).  \citet{rosenbaum2010designof} analyzed this data
 by matching each of the 21  welders to one of the 26 controls.  Here we  replicate
 this matching approach and   consider the impact of covariate-adaptive
 randomization inference in place of the uniform randomization inference typically
 used for pair-matched designs.
 
Following Rosenbaum, we estimate the  propensity  score using logistic regression
against the three covariates discussed above and match each welder to a control 
using a  robust Mahalanobis distance with a propensity score caliper equal
to half the standard deviation of the fitted propensity score values across the
dataset. Since it is not possible to match all 21 welders within  the caliper,
a  soft caliper is used in which violations of the caliper are  penalized  linearly.

As noted by Rosenbaum, pair matching  cannot remove all confounding in this 
data, particularly because only six of the 27 controls are removed by the 
matching process. Table \ref{tab:welder_pairs} shows the pair differences on each
of the three matching variables and the estimated propensity score, with
average discrepancies at  the bottom.   
Note that although the average discrepancies are small on all variables (the average 
difference in age is  only 1.5  years), in the  large majority of pairs the 
treated unit has a slightly 
higher propensity score, several with differences larger than 0.1.  In contrast
to standard uniform randomization inference, covariate-adaptive randomization
inference will  pay  attention to  these differences.

Using the difference in mean DNA-protein linkage between welders and matched
controls as a test statistic, we contrast uniform randomization inference 
structured by the matched
pairs with covariate-adaptive inference based on the estimated propensity score.
The density plot in Figure \ref{fig:welder_density} illustrates the difference between the two null
distributions.  As the consistent positive sign on the propensity score discrepancies would suggest, the
covariate-adaptive randomization distribution is shifted to the right of the uniform randomization distribution,
showing how accounting for residual propensity score differences leads us to expect larger values even under the null.

The data was collected based on a hypothesis that welders experience elevated
 levels  of genetic damage compared to controls, so we conduct a one-sided
 test of the  Fisher sharp null hypothesis of no difference in DNA-protein
 linkage due to welder status.  Using the null distributions shown above,
 this corresponds to calculating the proportion of null draws that exceed
 the observed value of  0.64.  For the uniform distribution this is 
0.015, and for the covariate-adaptive 
distribution it is 0.029.  The 
covariate-adaptive distribution adjusts for the residual propensity score
differences in the  pairs, which are biased towards treatment for the welders
who actually experienced treatment; as such, it recognizes that part of the
apparent treatment effect in the  uniform test is likely a result of bias and
concludes that the weight of evidence  in   favor of a treatment is weaker.
 While  both
tests are significant at the  0.05 level in this case, if researchers had been 
interested in effects of both signs and had  conducted  a two-sided  test the
covariate-adaptive test does not reject in this case.  As such, failure to account
properly for propensity score discrepancies in the two-sided case leads to 
an anticonservative result in which the null hypothesis is rejected when most properly 
it should not be.  When a sensitivity analysis is run, the uniform analysis is
similarly more optimistic about evidence for a treatment effect, reporting
similar qualitative results for $\Gamma$ up to 1.20, while the covariate-adaptive
analysis allows a maximum value of only 1.11.

\subsection{Right-heart catheterization and death after surgery}

We next conduct covariate-adaptive randomization inference in the right-heart
catheterization data of \citet{connors1996effectiveness}.  In this study
the effectiveness of the right-heart catheterization (RHC) procedure for improving
outcomes of critically ill patients was assessed by matching patients receiving
the procedure to those not receiving it and comparing mortality rates.  We 
follow the original authors by fitting a propensity score on 31 variables 
measured at intake and forming matches only between patients with propensity 
scores that differ by no more than 0.03.  Unlike the authors, who use a greedy
matching procedure, we use an optimal matching procedure that minimizes a
robust Mahalanobis distance formed from the 31 variables in the propensity score
model within propensity score calipers, and also matches
 exactly on primary disease category.  Not all of the 2,184 RHC patients
can be matched to distinct controls within the caliper, and in this case we
exclude the minimal number of RHC patients necessary for the caliper to be 
respected.  This leaves us with 1,538 matched pairs, which is still a 
substantial improvement on the match conducted by the authors, which had only
1,008 pairs.  Table  \ref{tab:balance_rhc} summarizes the effectiveness of the
match in removing bias on the observed pre-treatment variables.  The numbers
shown are standardized treatment-control differences in means for each variable,
both before (first column) and after (second column) matching.  Only the variables with
the 25 largest pre-matching absolute standardized differences are shown (but
post-matching standardized differences remain below 0.04 for all variables not
shown).  Clearly, although large differences between groups are present in the 
raw data for numerous  variables, notably APACHE risk score (aps1), mean blood
pressure (meanbp1), and P/F oxygen ratio (pafi1), matching transforms these
into small differences.  Of special note is the row showing balance on the
estimated propensity score, which shows a reduction from a standardized 
difference exceeding 1 to a value of only 0.04. 

To assess the role of the caliper in influencing the study's results, we 
construct an alternative matched design using a generalized version
of  optimal subset matching
\citep{rosenbaum2012optimal}, as implemented in the R package 
\texttt{rcbsubset}. Instead of using a propensity score caliper, this 
match enforces exact balance on ventiles of the propensity score,
ensuring close balance on the propensity score without 
imposing hard restrictions on the propensity score discrepancy
within a matched pair.  Treated units are excluded from the match
if their inclusion induces the average Mahalanobis distance across
matched pairs to cross a specific threshold given by a tuning parameter;
we chose a value of the tuning parameter that induces a similar overall
sample size as in the calipered match (1,507 matched pairs).   Balance on important pre-treatment
variables is similar to that in the caliper match, as shown in the third column
of Table \ref{tab:balance_rhc}.  However, there are important differences between the 
two matches at the level of the pairs, as shown in Table \ref{tab:pair_rhc}; the caliper
match achieves much greater similarity of propensity scores within pairs, at 
the cost of achieving slightly reduced similarity on a range of other variables.
In practice the caliper match is likely to be more attractive, since it achieves
 a major improvement in propensity score control at the cost of relatively minor
 changes in other variables, but both matches are Pareto optimal in the sense of
 \citet{pimentel2020optimaltradeoffs}.

The outcome of interest in this study is patient death within thirty days.
For each match we conduct four outcome analyses: uniform and covariate-adaptive
 randomization inference for the difference in means, and uniform and
 covariate-adaptive randomization inference for a regression-adjusted 
 difference-in-means statistic \citep{rosenbaum2002covariance}.
Although the outcome is binary, the mortality rate is sufficiently
 high that an ordinary least-squares fit is reasonable.  Table \ref{tab:outcome_rhc} 
 summarizes the results of the analysis, providing point estimates,
 p-values,
 and estimated confidence intervals from all four approaches, with
  the point estimates for the covariate-adaptive procedures 
  shifted by the mean of the covariate-adaptive null distribution given in formula (\ref{eqn:nullmean}).  In addition,
results from a sensitivity analysis are reported, giving the largest values of $\Gamma$ consistent with a rejection of the null
hypothesis (the ``threshold" $\Gamma$).

  Results are generally consistent with an increased risk of 
  mortality associated with right-heart catheterization, with all confidence intervals contained in the range 3\%-11\%.  As in the welders data, the covariate-adaptive procedure reports a treatment effect slightly smaller than the uniform procedure; however, in this case the magnitude of the difference is small compared to the overall size of the effect so the methods all agree qualitatively.  Notice also that the regression-adjusted test statistics have narrower confidence intervals than the raw risk differences, in line with principles described in \citet{rosenbaum2002covariance}. The regression-adjusted analyses
  lead to smaller estimates of the treatment effect which explains their greater sensitivity
  to unmeasured bias.  The caliper matching has more stable threshold $\Gamma$ values across analyses, 
  consistent with the tighter control of the propensity score it achieves.
For the subset match, the covariate adaptive procedure tends to increase the threshold $\Gamma$ 
despite having smaller treatment effects; this is because it also lowers the variance of the estimator
and leads to tighter confidence intervals.

\section{Discussion}
\label{sec:discussion}
  Uniform randomization inference for matched studies relies, often at least in part implicitly, on a model in which unobserved confounders are absent, in which propensity scores are matched exactly, and (depending on the sampling model) in which the matched sets selected are conditionally independent of the original treatment vector.  Substantial failures of these assumptions lead to substantial problems with Type I error control and require some form of correction.  Covariate-adaptive randomization provides such correction by altering permutation probabilities in the randomization test based on estimated propensity scores.  Relative to the na\"{i}ve analysis for the difference-in-means statistic, covariate-adaptive randomization tends to restore approximate control of Type I error in ma	ny settings and constitutes an attractive option alongside approaches based on regression-adjusted test statistics and matching calipers.
  
 The applied examples show that covariate-adaptive randomization inference need not change the qualitative results of an observational study, especially when the study designer has given careful attention to propensity score discrepancies.  In these settings the covariate-adaptive procedure can still be a productive robustness check to help build confidence that lingering propensity score discrepancies are not corrupting the study's key findings.
 
 Covariate-adaptive randomization inference builds on a quickly-growing literature that explores using estimated propensity scores to structure permutation tests.  \citet{rosenbaum1984conditional} recommended permuting treatment assignment conditional on the sufficient statistic for a propensity score fit, a closely related idea which works very well for settings with only one or two discrete covariates with a limited number of categories.  In a graduate dissertation \citet{baiocchi2011methodologies} briefly proposed permuting treatment assignments within matched pairs in a manner similar to that described above, although with a slightly different distribution based on a ratio of propensities on the  probability scale rather than the odds scale.
  More recently, the conditional permutation test of \citet{berrett2020conditional} proposes using permutations of observed variables based on an estimated conditional distribution for the purposes of testing conditional independence.  Covariate-adaptive randomization tests in matched designs may be viewed as as modified versions of this conditional permutation test, a connection exploited repeatedly in the results of Section \ref{sec:pscore_error}; the structure of the matched sets resolves many of the challenges that arise in  \citet{berrett2020conditional}'s construction, such as the question of how to sample from the permutation distribution.  \citet{shaikh2021randomization} also use permutations of observations based on estimated propensity scores to conduct causal inference in an observational study and obtain even stronger guarantees about large-sample performance by leveraging specific focus on a setting in which only one unit receives treatment.

Covariate-adaptive inference raises several interesting future directions for theory and methods development.  First, the theoretical guarantee given in this work is limited both by a strong restriction on the relative sample sizes of the analysis sample and of a hypothetical pilot sample used to fit the propensity score.  Clarifying whether this assumption can be relaxed and whether same-sample estimation of some kind, such as cross-fitting, could be applied instead would be a valuable contribution.  For the difference-in-means estimator, one likely challenge is that the bias of the covariate-adaptive randomization distribution based on the estimated propensity score may not shrink to zero at a strictly faster rate than the variance when propensity-fitting and analysis samples are of a similar order, so that preserving valid inference may require widening confidence intervals or reducing the significance threshold $\alpha$ to account for the bias. 

Secondly, covariate-adaptive randomization inference offers opportunities to develop stronger model-based guidance about the relative importance of the propensity score and the prognostic score in the construction of matched designs.  Currently many design objectives proposed as bases for constructing matched sets, including mean balance, caliper matching, and multivariate distance matching have a heuristic element in the sense that it is not clear for which class of overall models for treatment and outcome they provide optimal results.  While \citet{kallus2020generalized} lays out a helpful overarching framework under sampling-based inference, unifying descriptions are not present for randomization-based inferences in matching.  Covariate-adaptive inference provides important progress towards this goal by clarifying the impact of inexact propensity score matching on the operating characteristics of the ultimate estimate and associated test.  Power analysis and design sensitivity calculations \citep{rosenbaum2010designsensitivity} based on the covariate-adaptive model, if developed, would provide valuable design-stage insight about how to properly use the propensity score in constructing the matches, and could provide more definitive guidance to users selecting among many Pareto optimal matches \citep{pimentel2020optimaltradeoffs}.

Thirdly, the evidence provided by our simulation study suggests that $Z$-dependence may play a limited but non-trivial role in 
failures of type I error control for matched randomization tests.  While from a technical standpoint it may be sufficient to consider the matched pairs as the basic units of interest, sampled from some distribution, this makes it difficult to develop formal understanding of how choices about matching influence inference.  In particular, some statistical model on the original subjects of the study is needed to obtain the strong model-based guidance about the proper construction of matched sets discussed in the previous paragraph.  $Z$-dependence differs fundamentally from propensity score discrepancies in the sense that it alters the support of the randomization distribution, not just the values of nonzero sampling probabilities, and new methods, perhaps algorithmic approaches to restricting randomization distributions, must be constructed to understand and control it.

{
\singlespace
\bibliographystyle{asa}
\bibliography{../../../bibmaster/matching}}

\begin{table}[ht]
\small
\centering
\begin{tabular}{rrrrr}
  \hline
 & African-American & Age & Smoker & PS \\ 
  \hline
Pair 1 & 0.00 & -6.00 & 0.00 & 0.12 \\ 
  Pair 2 & 0.00 & -3.00 & 0.00 & 0.05 \\ 
  Pair 3 & 0.00 & 4.00 & 0.00 & 0.05 \\ 
  Pair 4 & 0.00 & -8.00 & 0.00 & 0.16 \\ 
  Pair 5 & 0.00 & 0.00 & 0.00 & 0.00 \\ 
  Pair 6 & 0.00 & 3.00 & 0.00 & -0.06 \\ 
  Pair 7 & 0.00 & -3.00 & 0.00 & 0.06 \\ 
  Pair 8 & 0.00 & -2.00 & 0.00 & 0.17 \\ 
  Pair 9 & 0.00 & 3.00 & 0.00 & -0.06 \\ 
  Pair 10 & 0.00 & -5.00 & 0.00 & 0.06 \\ 
  Pair 11 & 0.00 & -3.00 & 0.00 & 0.06 \\ 
  Pair 12 & 0.00 & -5.00 & 0.00 & 0.10 \\ 
  Pair 13 & 0.00 & -5.00 & 0.00 & 0.11 \\ 
  Pair 14 & 0.00 & -5.00 & 0.00 & 0.10 \\ 
  Pair 15 & 0.00 & 0.00 & 0.00 & 0.00 \\ 
  Pair 16 & 0.00 & -4.00 & 0.00 & 0.08 \\ 
  Pair 17 & 0.00 & 1.00 & 0.00 & -0.02 \\ 
  Pair 18 & 0.00 & 1.00 & 0.00 & -0.01 \\ 
  Pair 19 & -1.00 & 7.00 & 0.00 & 0.04 \\ 
  Pair 20 & 0.00 & -5.00 & 0.00 & 0.10 \\ 
  Pair 21 & 0.00 & 4.00 & 0.00 & 0.05 \\ 
  Average & -0.05 & -1.48 & 0.00 & 0.06 \\ 
   \hline
\end{tabular}
\caption{\small Treated- control differences in matched pairs in the welders dataset.}
\label{tab:welder_pairs}
\end{table}

 \begin{table}[ht]
 \small
\centering
\begin{tabular}{rrrr}
  \hline
 & Before Matching & Caliper Match & Subset Match \\ 
  \hline
Est. Propensity Score & 1.254 & 0.002 & 0.033 \\ 
  Acute Physiology Score & 0.501 & 0.016 & 0.068 \\ 
  Mean Blood Pressure & -0.455 & -0.026 & -0.058 \\ 
  PaO2/(.01*FiO2) & -0.433 & 0.001 & 0.030 \\ 
  Serum Creatinine & 0.270 & 0.021 & 0.044 \\ 
  Hematocrit & -0.269 & 0.001 & -0.008 \\ 
  Weight (kg) & 0.256 & 0.002 & -0.017 \\ 
  PaCO2 & -0.249 & 0.004 & -0.013 \\ 
  MOSF w/Sepsis (secondary) & 0.230 & -0.036 & -0.014 \\ 
  Albumin & -0.230 & -0.018 & -0.021 \\ 
  Predicted Survival & -0.198 & -0.031 & -0.049 \\ 
  MOSF w/Sepsis (primary) & 0.172 & 0.000 & 0.000 \\ 
  Respiration Rate & -0.165 & -0.002 & -0.011 \\ 
  Heart Rate & 0.147 & 0.027 & 0.000 \\ 
  Bilirubin & 0.145 & -0.013 & 0.010 \\ 
  Heart disease & 0.139 & 0.011 & 0.014 \\ 
  MOSF w/Malignancy (secondary) & -0.135 & 0.009 & -0.017 \\ 
  Respiratory Disease & -0.128 & -0.007 & -0.020 \\ 
  Serum pH & -0.120 & -0.002 & -0.044 \\ 
  Missing ADL Score & 0.117 & -0.001 & 0.005 \\ 
  SUPPORT Coma Score & -0.110 & 0.045 & 0.022 \\ 
  Neurological disease & -0.108 & 0.006 & 0.001 \\ 
  Serum Sodium & -0.092 & -0.011 & -0.001 \\ 
  Years of Education & 0.091 & 0.022 & 0.024 \\ 
  Sepsis & 0.091 & 0.012 & 0.017 \\    \hline
\end{tabular}
  \caption{\small Standardized differences in means before matching and under two different matches for the 25 variables with largest initial imbalance in the right-heart catheterization dataset.}
 \label{tab:balance_rhc}
\end{table}
 
 \begin{table}[ht]
 \small
\centering
\begin{tabular}{rrr}
  \hline
& Caliper Match & Subset Match \\ 
Avg. Propensity Score Discrepancy & 0.016 & 0.166 \\ 
  Max Propensity Score Discrepancy & 0.030 & 0.672 \\ 
  Prop. Matched Exactly, Male Sex & 0.351 & 0.297 \\ 
  Avg. Discrepancy, Age & 13.469 & 10.892 \\ 
  Avg. Discrepancy, Mean Blood Pressure & 27.809 & 25.506 \\ 
  Avg. Discrepancy, Heart Rate & 32.677 & 28.858 \\ 
  Avg. Discrepancy, Respiration Rate & 11.614 & 10.426 \\ 
  Avg. Discrepancy, PaO2/(.01*FiO2) & 76.651 & 75.882 \\ 
 \hline
\end{tabular}
  \caption{\small Matched pair quality under two different matches for selected variables the right-heart catheterization dataset.}
 \label{tab:pair_rhc}
\end{table}

  \begin{table}[ht]
  \small
\centering
\begin{tabular}{rrr|rr}
  \hline
  &\multicolumn{2}{c|}{Caliper match} & \multicolumn{2}{c}{Subset match}\\
& Uniform & Covariate-Adaptive & Uniform  & Covariate-Adaptive  \\ 
  \hline
Risk difference & 0.074 & 0.073 & 0.068 & 0.062 \\ 
  Lower conf. limit & 0.043 & 0.041 & 0.038 & 0.034 \\ 
  Upper conf. limit & 0.106 & 0.104 & 0.098 & 0.089 \\ 
  Threshold $\Gamma$ & 1.26 & 1.26 & 1.24 & 1.26 \\ 
  Risk difference, OLS-adjusted & 0.051 & 0.051 & 0.055 & 0.054 \\ 
  Lower conf. limit, OLS-adjusted & 0.022 & 0.022 & 0.027 & 0.028 \\ 
  Upper conf. limit, OLS-adjusted & 0.081 & 0.081 & 0.083 & 0.079 \\ 
  Threshold $\Gamma$, OLS-adjusted & 1.12 & 1.12 & 1.15& 1.19 \\ 
   \hline
\end{tabular}
  \caption{\small Outcome analysis for 30-day mortality in the right-heart catheterization data.  Results are shown for both the caliper match and the subset match, with and without regression adjustment, and using both uniform and covariate-adaptive inference.}
  \label{tab:outcome_rhc}
\end{table}

\begin{figure}
\hspace{-2em}\includegraphics{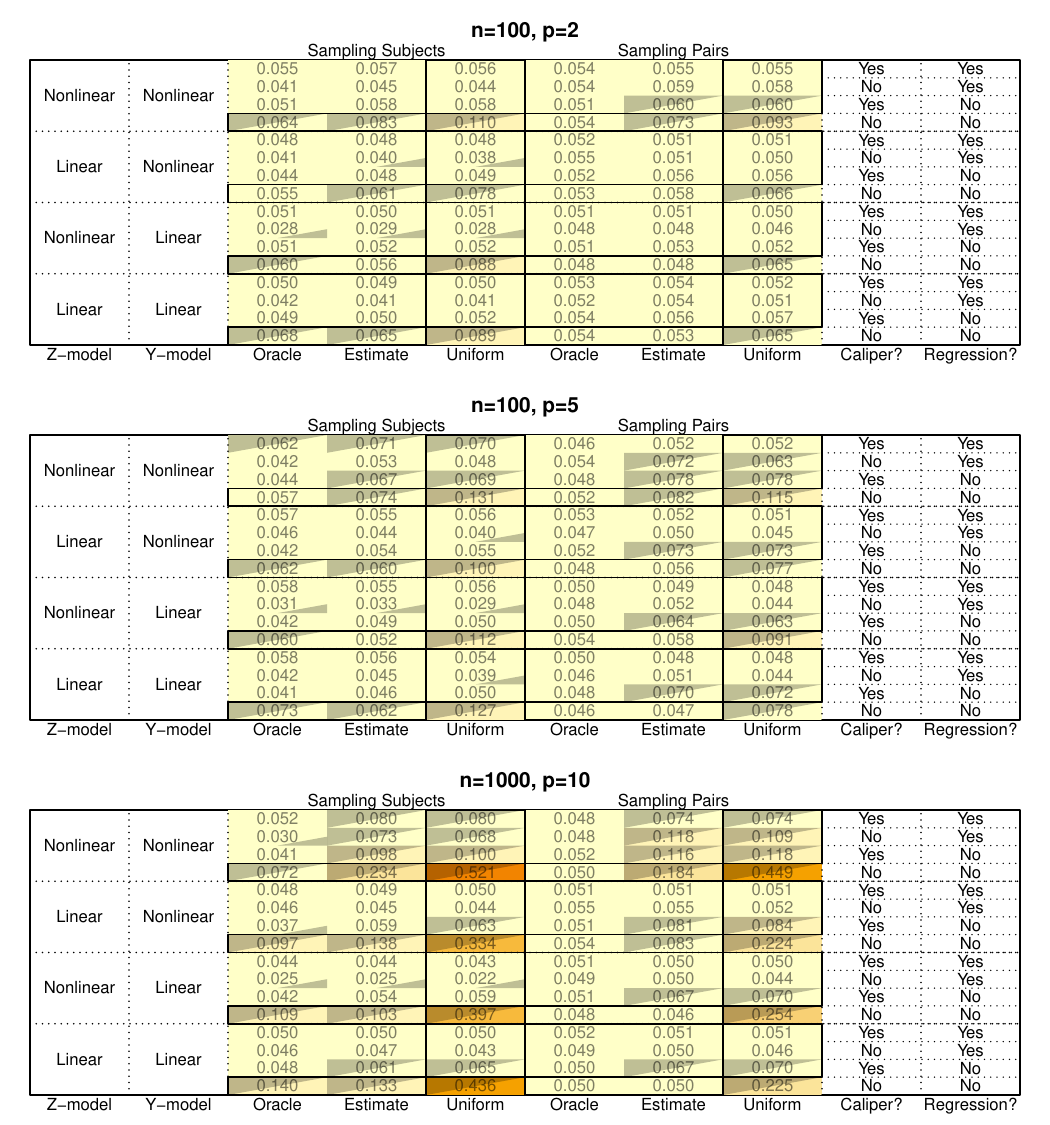}

\caption{\small Type I error results for uniform and covariate-adaptive inference across multiple simulation settings.   Each of the three tables corresponds to a separate dataset size.  Within each table the first three columns  contrast three inferential approaches when the subjects are sampled independently prior to matching; the last three columns do the same comparison when treatment assignments within matched sets are sampled independently after matching.  The rows of the table demonstrate different combinations of calipers and regression adjustment and correct or incorrect specification of treatment and outcome models.  Numbers give type I error rates with colors associated to their magnitude; triangles indicate that a one-sample z-test rejected the null hypothesis that the error rate was 0.05 (under a Bonferroni correction scaled to the number of results across the entire figure), with a large upper triangle indicating a positive z-statistic and a small lower triangle indicating a negative z-statistic. }
\label{fig:sim_heatmaps}
\end{figure}

\begin{figure}
\includegraphics{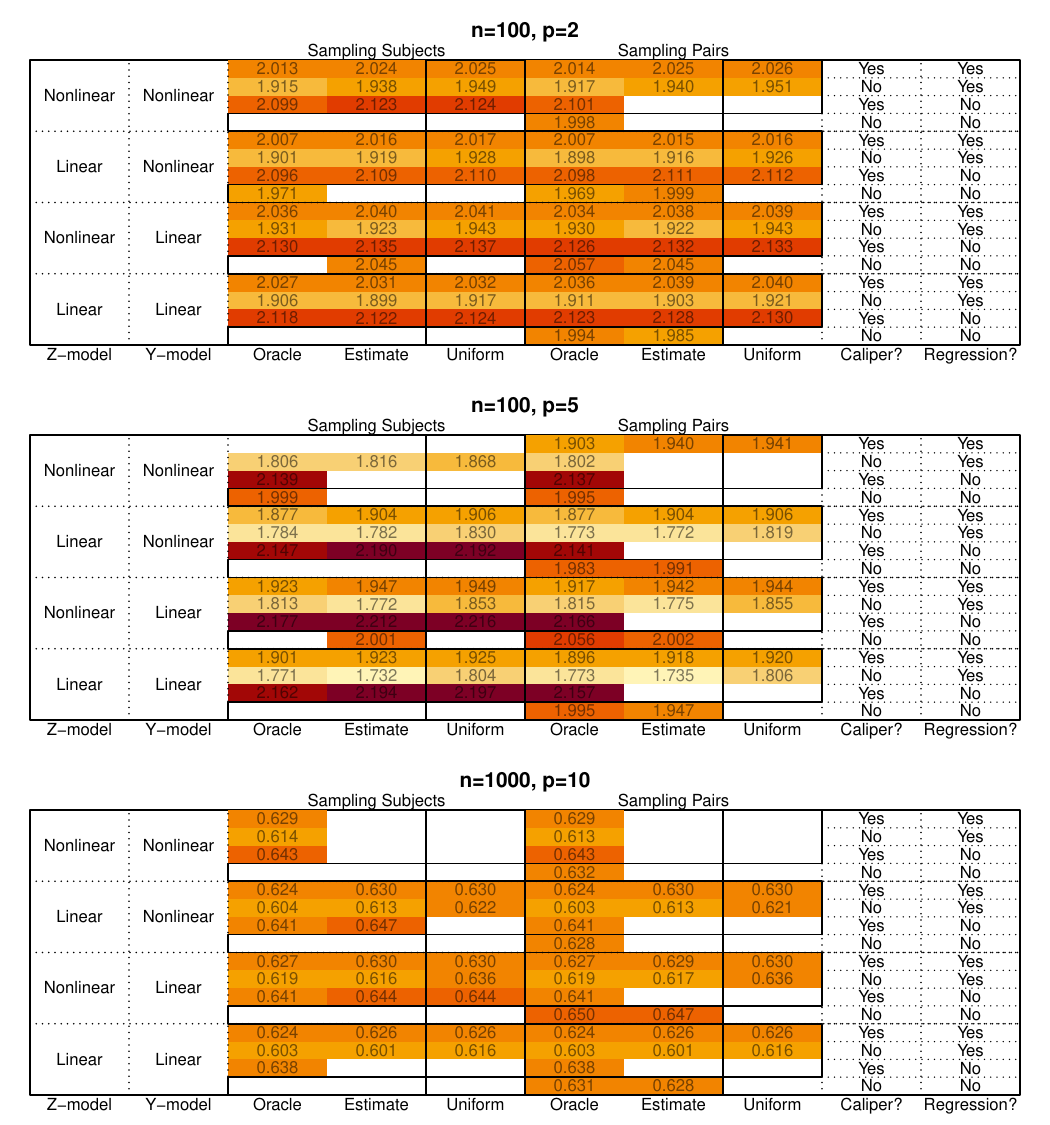}

\caption{\small Average confidence interval length for uniform and covariate-adaptive inference across multiple simulation settings for cases with approximate control of type I error at 0.05 or less.   Within each table the first three columns  contrast three inferential approaches when the subjects are sampled independently prior to matching; the last three columns do the same comparison when treatment assignments within matched sets are sampled independently after matching.  The rows of the table demonstrate different combinations of calipers and regression adjustment and correct or incorrect specification of treatment and outcome models. }
\label{fig:sim_heatmaps_ci}
\end{figure}

\begin{figure}
\centering
\includegraphics[scale=0.4]{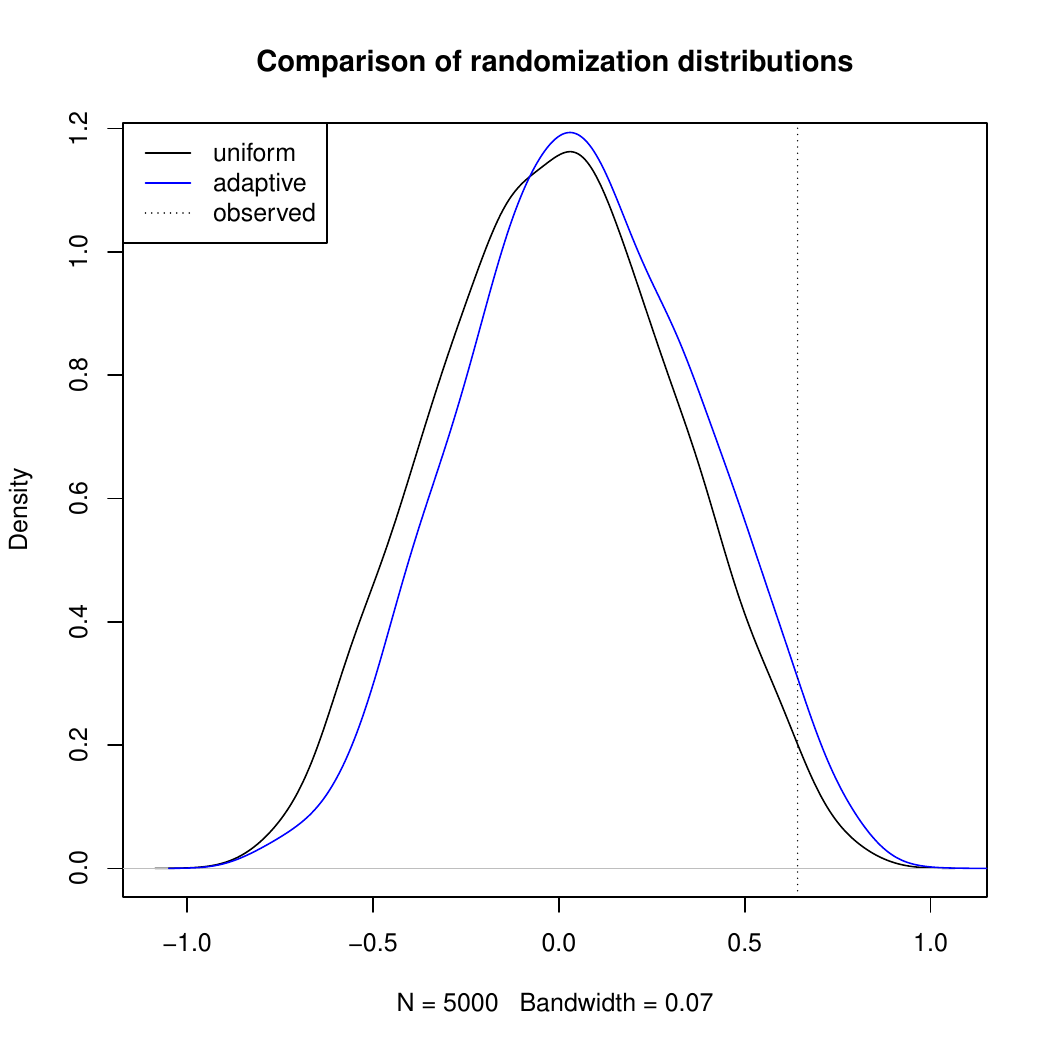}
\caption{\small Smoothed densities for the uniform randomization distribution of the difference in means statistic and the covariate-adaptive randomization distribution in the welders dataset, with the value of the observed statistic.}
\label{fig:welder_density}
\end{figure}

\setcounter{section}{0}
\renewcommand{\thesection}{\Alph{section}}
\section{Appendix}

\subsection{Proof of Theorem 2}
\label{subsec:logit_proof}

\begin{proof}
We use Pinsker's inequality to bound the quantity $d_{TV}(P_\lambda,P_{\widehat{\lambda}})$ by the square root of the sum of the Kullback-Leibler divergences between true and estimated conditional distributions of treatment for each unit in the study, then use a Taylor expansion to show that this quantity converges in probability to zero.  In combination with Theorem 1,
this suffices for the result.  

We let $\widehat{\beta}$ represent the maximum likelihood estimate of $\beta$ obtained from the size-$N$ pilot sample.
\begin{align*}
d_{TV}^2(P_\lambda,P_{\widehat{\lambda}}) &\leq \frac{1}{2}\sum^n_{i=1}d_{KL}\left(P_{\lambda_i}, P_{\widehat{\lambda}_i}\right)\\
&= \frac{1}{2}\sum^n_{i=1}E\log\left\{\left(\frac{1 + \exp[-\widehat{\beta}^Tx_i ]}{1 + \exp[-{\beta}^Tx_i ]}\right)^{Z_i}\left(\frac{1 + \exp[\widehat{\beta}^Tx_i ]}{1 + \exp[{\beta}^Tx_i ]}\right)^{1-Z_i} \right\}\\
&=  \frac{1}{2}\sum^n_{i=1}E\left\{Z_i\log\left(\frac{1 + \exp[-\widehat{\beta}^Tx_i ]}{1 + \exp[-{\beta}^Tx_i ]}\right)  + (1- Z_i)\log\left(\frac{1 + \exp[\widehat{\beta}^Tx_i ]}{1 + \exp[{\beta}^Tx_i ]}\right)\right\}\\
&= \frac{1}{2}\sum^n_{i=1}E\left\{\frac{1}{1 + \exp[-\beta^Tx_i]}\log\left(\frac{1 + \exp[-\widehat{\beta}^Tx_i ]}{1 + \exp[-{\beta}^Tx_i ]}\right)  + \frac{1}{1 + \exp[\beta^Tx_i]}\log\left(\frac{1 + \exp[\widehat{\beta}^Tx_i ]}{1 + \exp[{\beta}^Tx_i ]}\right)\right\}\\
&=\frac{1}{2}\sum^n_{i=1}\left\{C_{1i}f(-\widehat{\beta}^Tx_i) + C_{2i}f(\widehat{\beta}^Tx_i) + C_{3i}\right\} = g(\widehat{\beta}).
\end{align*} 
where 
\begin{align*}
C_{1i} &=1/(1 + \exp[-\beta^Tx_i]) \in (0,1)\\
C_{2i} &=1/(1 + \exp[\beta^Tx_i] ) \in (0,1)\\
C_{3i} &= -C_{1i}\log(1 + \exp[-\beta^Tx_i]) - C_{2i}\log(1 + \exp[\beta^Tx_i]) \\
f(b) &= \log(1 + \exp(b))
\end{align*}

We use a first-order multivariate Taylor expansion to understand the behavior of the vector-valued nonlinear function $g$ around the vector $\beta^TX$.
First note that 
\begin{align*}
f'(b) &= \frac{\exp(b)}{1 + \exp(b)} = \frac{1}{1 + \exp(-b)}\\
f''(b) &= \frac{\exp(-b)}{(1 + \exp(-b))^2} = \frac{1}{\exp(b) + 2 + \exp(-b)} \leq \frac{1}{2}\\
\end{align*}
so:
\begin{align*}
\nabla g(\beta^T{X}) &= \frac{1}{2}\sum^n_{i=1}\left\{-C_{1i}f'(-\beta^Tx_i) +C_{2i}f'(\beta^Tx_i)   \right\}x_i =  \frac{1}{2}\sum^n_{i=1}\left\{-C_{1i}C_{2i} +C_{2i}C_{i1}   \right\}x_i = \mathbf{0}.\\
H_g(\beta^TX) &= \frac{1}{2}\sum^n_{i=1}\left\{C_{1i}f''(-\beta^Tx_i) +C_{2i}f''(\beta^Tx_i)   \right\}x_ix_i^T 
\end{align*}

Taylor's theorem gives us the following:
\begin{align*}
 g(\widehat{\beta})&= g(\beta) +\nabla g(\beta^T{X})^T (\widehat{\beta} - \beta) + 
R_2\\
&= 0 + 0 + R_2 \end{align*}
Furthermore, using the Lagrange bound we can bound the second-order remainder $R_2$  (letting $|A|$ represent a matrix with each of its elements replaced by its absolute value).  Specifically, for some $\widetilde{\beta}$ near $\beta$, we have the following:
\begin{align}
g(\widehat{\beta}) &\leq (\widehat{\beta} - \beta)^T|H_g(\widetilde{\beta})|(\widehat{\beta} - \beta)\nonumber\\
&\leq \frac{1}{2}\sum^n_{i=1}\left|C_{1i}f''(-\beta^Tx_i) +C_{2i}f''(\beta^Tx_i)   \right|(\widehat{\beta} - \beta)^T|x_ix_i^T|(\widehat{\beta} - \beta)\nonumber\\
&\leq  \frac{1}{2}\sum^n_{i=1}\left(\frac{1}{2} + \frac{1}{2}  \right)(\widehat{\beta} - \beta)^T|x_ix_i^T|(\widehat{\beta} - \beta)\nonumber\\
&=  \frac{1}{2}\sum^n_{i=1}\frac{\left[\sqrt{N}(\widehat{\beta} - \beta)^T|x_i|\right]^2}{N} \label{eqn:g_summation}
\end{align}
Since $\widehat{\beta}$ is a maximum likelihood estimate fitted on a sample of size $N$, for any $x_i$  the quantity $\sqrt{N}(\widehat{\beta} - \beta)^Tx_i$ converges in distribution to a normal random variable with mean zero and variance $x_i^T\Sigma x_i$ for some variance-covariance matrix $\Sigma$, and by the continuous mapping theorem, each term in summation (\ref{eqn:g_summation}) converges in distribution to 
\[
\frac{Wx_i^T\Sigma x_i}{N}
\] where $W \sim \chi^2_1$.  Since $x_i$ is chosen from compact support, there is a uniform upper bound $\widetilde{\sigma}^2$ on these variances $x_i^T\Sigma x_i$ for all $i = 1, \ldots, n$.  Then for any $\epsilon > 0$ and any $\delta > 0$, we may choose $N_{\epsilon, \delta}$ sufficiently large to ensure that  for all $N \geq N_{\epsilon, \delta}$,
\[
 P\left(W > \left( \frac{N}{\widetilde{\sigma}^2}\right)\delta\right) < \epsilon
\]
Thus each term in summation (\ref{eqn:g_summation}) is $O_p\left(\frac{\widetilde{\sigma}^2}{N}\right)$, and in turn the entire summation is $O_p\left(\frac{n\widetilde{\sigma}^2}{N}\right)$, and it square root is $O_p\left(\sqrt{\frac{n\widetilde{\sigma}^2}{N}}\right)$.  
Since $\lim_{n\longrightarrow \infty} n/N = 0$, this is sufficient for $\sqrt{g(\widehat{\beta})}$ to converge in probability to zero.

To complete the proof, fix any $\epsilon > 0$.   We now establish that there exists $N_\epsilon$ such that for all $N > N_\epsilon$, 
\[
P(p_{\widehat{\lambda}_N,n} \leq \alpha) - \alpha < \epsilon.
\]
which suffices for the desired result.
To do this, note that the convergence of $d_{TV}(P_\lambda,P_{\widehat{\lambda}})$ to zero in probability implies that there exists $N'$ such that for all $N > N'$,
\[
P(d_{TV}(P_\lambda,P_{\widehat{\lambda}}) > \epsilon/2) < \epsilon/2. 
\]
Then when $N > N'$, by  Theorem 1:
\begin{align*}
P(p_{\widehat{\lambda}_N,n} \leq \alpha) - \alpha &\leq d_{TV}(P_\lambda,P_{\widehat{\lambda}})\\
&= 1\{d_{TV}(P_\lambda,P_{\widehat{\lambda}}) \leq \epsilon/2\} d_{TV}(P_\lambda,P_{\widehat{\lambda}}) + 1\{ d_{TV}(P_\lambda,P_{\widehat{\lambda}}) > \epsilon/2\} d_{TV}(P_\lambda,P_{\widehat{\lambda}})\\
\text{so}&\\
E_{\widehat{\lambda}}\left[P(p_{\widehat{\lambda}_N,n} \leq \alpha) - \alpha  \right]&\leq E_{\widehat{\lambda}}\left[1\{d_{TV}(P_\lambda,P_{\widehat{\lambda}}) \leq \epsilon/2\} d_{TV}(P_\lambda,P_{\widehat{\lambda}}) \right]+ E_{\widehat{\lambda}}\left[1\{ d_{TV}(P_\lambda,P_{\widehat{\lambda}}) > \epsilon/2\} d_{TV}(P_\lambda,P_{\widehat{\lambda}})\right]\\
P(p_{\widehat{\lambda}_N,n} \leq \alpha) - \alpha &\leq \epsilon/2 + P_{\widehat{\lambda}}(d_{TV}(P_\lambda,P_{\widehat{\lambda}}) > \epsilon/2) \\
&< \epsilon/2 + \epsilon/2 = \epsilon
\end{align*}
as desired.  As indicated, expections are taken with respect to the sampling distribution of the estimated propensity scores, which does not affect the left-hand side because it is a constant. The second-to-last line follows from the law of total probability and the upper bound of 1 on the total variation distance.
\end{proof}

\subsection{Proof of Theorem  3}
\label{subsec:sens_proof}

We first define some terms, following \citet{rosenbaum2002observational}[\S 2].   Consider an arbitrary function $g: \mathbb{R}^{n'} \times \mathbb{R}^{n'} \longrightarrow \mathbb{R}$. We say that $g$ is \emph{invariant} if the value of $g(\mathbf{a}, \mathbf{b})$ remains unchanged when the two arguments are subjected to the same permutation, and we define certain invariant functions as \emph{arrangement-increasing}.  

\begin{definition}
\label{def:AI}
For any $\mathbf{x} \in \mathbb{R}^{n'}$ with indices partitioned into $K$ strata, let $\mathbf{x}
_{kij}$ be the vector obtained by interchanging the $i$th and $j$th entries for stratum $k$ within vector $\mathbf{x}$.  An invariant function $g: \mathbb{R}^{n'} \times \mathbb{R}^{n'} \longrightarrow \mathbb{R}$ is \emph{arrangement-increasing} if and only if $g(\mathbf{a}, \mathbf{b}_{kij}) \geq g(\mathbf{a}, \mathbf{b})$ whenever
\[
(a_{ki} - a_{kj})(b_{ki} - b_{kj}) \leq 0.
\]
\end{definition}

Intuitively, a function is arrangement-increasing if it takes on larger values as its two arguments are sorted into a more similar ordering; for more properties and examples see \citet{hollander1977functions} and \citet[ch. 2]{rosenbaum2002observational}.   If we abuse notation slightly to rewrite the sum statistic $T(\mathbf{Z}_\mathcal{M}, \mathbf{Y}_\mathcal{M})$ as a function $T(\mathbf{Z}_\mathcal{M}, \mathbf{Y}_\mathcal{M})$ of $\mathbf{f}(\mathbf{Y})$, then $T$ is arrangement-increasing as shown in \citet{rosenbaum2002observational,rosenbaum2007sensitivity}.
As an immediate consequence, the indicator function $I_C(\mathbf{Z}, \mathbf{f}(\mathbf{Y})) = 1\left\{T(\mathbf{Z}_\mathcal{M}, \mathbf{f}(\mathbf{Y}_\mathcal{M})) > C\right\}$ is also arrangement-increasing for all $C$.
Assuming model (6) from the main manuscript 
for treatment assignment, take any arrangement-increasing $h(\mathbf{f}, \mathbf{u})$ and define:
\[
\omega(\mathbf{f}, \mathbf{u}) = E_\mathbf{Z}\{h(\mathbf{Z}, \mathbf{f})\} = \sum_{\mathbf{z} \in \boldsymbol{\Omega}}h(\mathbf{z}, \mathbf{f})P(\mathbf{Z} = \mathbf{z}) = \mathbf{z}) = \frac{\sum_{\mathbf{z} \in \boldsymbol{\Omega}}h(\mathbf{z}, \mathbf{f})\exp\left\{\mathbf{z}^T(\boldsymbol{\kappa} + \gamma \mathbf{u}) \right\}}{\sum_{\mathbf{b} \in \boldsymbol{\Omega}}\exp\left\{\mathbf{b}^T(\boldsymbol{\kappa} + \gamma \mathbf{u}) \right\}}
\]
By choosing $h(\cdot)$ as an indicator function $I_C(\cdot)$ with $C$ equal to the critical value for our one-sided test, we obtain $\omega(\mathbf{f},\mathbf{u}) = \alpha_{adapt}(\mathbf{Y},\mathbf{u})$, so it will be sufficient to show that
\begin{align}
\min_{\mathbf{u}' \in U^-}\omega(\mathbf{f},\mathbf{u}') \leq\omega(\mathbf{f},\mathbf{u}) \leq \max_{\mathbf{u} \in U^+}\omega(\mathbf{f},\mathbf{u}).
\label{eqn:omegabound}
\end{align}
to complete the proof.  Note that other choices of $h(\cdot)$ allows similar bounding arguments for quantities such as the Hodges-Lehmann estimate.
%
%
%

To establish (\ref{eqn:omegabound}), we prove two lemmas, analogous to Propositions 17 and 18 of \citet{rosenbaum2002observational}.  The proof of Theorem 3
then follows from Rosenbaum's proof of his Proposition 19, substituting our lemmas for Propositions 17 and 18 as appropriate.

\begin{lemma}
\label{lem:omega_AI}
$\omega(\mathbf{p}, \mathbf{u})$ is an arrangement-increasing function.
\end{lemma}
\begin{proof}
We have chosen $h(\mathbf{z}, \mathbf{p})$ to be arrangement-increasing.  If in addition $g(\mathbf{z}, \mathbf{u})$ is arrangement-increasing, where
\[
g(\mathbf{z}, \mathbf{u}) = \frac{\exp\left\{\mathbf{z}^T(\boldsymbol{\kappa} + \gamma \mathbf{u}) \right\}}{\sum_{\mathbf{b} \in \boldsymbol{\Omega}}\exp\left\{\mathbf{b}^T(\boldsymbol{\kappa} + \gamma \mathbf{u}) \right\}}
\]
then by the composition theorem of \citet{hollander1977functions} $\omega(\mathbf{p}, \mathbf{u})$ is also arrangement-increasing, since it is a composition of $h$ and $g$ over a discrete uniform distribution.  Therefore the following direct demonstration that $g$ is arrangement-increasing completes the proof.  Choose any $\mathbf{z}, \mathbf{u}, k,i,j$ such that $(z_{ki} - z_{kj})(u_{ki} - u_{kj}) \leq 0$, and let $\mathbf{u}_{kij}$ as in Definition \ref{def:AI}.  The proof is trivial if $z_{ki} = z_{kj}$, so WLOG assume $z_{ki} = 1$ and $z_{kj} = 0$ (and therefore $u_{ki} \leq u_{kj}$).
\begin{align*}
\frac{g(\mathbf{z}, \mathbf{u}_{ij})}{g(\mathbf{z}, \mathbf{u})} = \frac{\exp\left\{\mathbf{z}^T(\boldsymbol{\kappa} + \gamma \mathbf{u}_{ij})\right\}}{\exp\left\{\mathbf{z}^T(\boldsymbol{\kappa} + \gamma \mathbf{u}_{})\right\}}
&= \prod_{k=1}^n \frac{\exp\left\{{z}_k(\kappa_k + \gamma (\mathbf{u}_{ij}))_k\right\}}{\exp\left\{{z}_k(\kappa_k + \gamma {u}_{k})\right\}} = \frac{\exp(\kappa_i + \gamma u_j)}{\exp(\kappa_i + \gamma {u}_i)} = \exp(\gamma(u_j - u_i)) \geq 0.
\end{align*}
so $g(\mathbf{z}, \mathbf{u}_{kij}) \geq g(\mathbf{z}, \mathbf{u})$.
\end{proof}

\begin{lemma}
\label{lem:omega_monotone}
Let $\mathbf{e}_{\ell i}$ be a vector containing a one in index $i$  of matched set $\ell$ and zeroes in all other indices.  For each fixed $\mathbf{f}, \mathbf{u}$, and $i$, $\omega(\mathbf{f}, \mathbf{u} + \delta\mathbf{e}_{\ell i})$ is monotone in $\delta$.
\end{lemma}
\begin{proof} Let $\boldsymbol{\Omega}_1$ be the set of vectors $\mathbf{b} \in \boldsymbol{\Omega}$ for which $b_i = 1$, and let $\boldsymbol{\Omega}_0 = \boldsymbol{\Omega} - \boldsymbol{\Omega}_1$.
\begin{align*}
\omega(\mathbf{f}, \mathbf{u}+ \delta\mathbf{e}_i) & = \frac{\sum_{\mathbf{z} \in \boldsymbol{\Omega}}h(\mathbf{z}, \mathbf{f})\exp\left\{\mathbf{z}^T(\boldsymbol{\kappa} + \gamma \mathbf{u} +\gamma \delta\mathbf{e}_i) \right\}}{\sum_{\mathbf{b} \in \boldsymbol{\Omega}}\exp\left\{\mathbf{b}^T(\boldsymbol{\kappa} + \gamma \mathbf{u} + \gamma\delta\mathbf{e}_i) \right\}} \\
&= \frac{\sum_{\mathbf{z} \in \boldsymbol{\Omega}_0}h(\mathbf{z}, \mathbf{f})\exp\left\{\mathbf{z}^T(\boldsymbol{\kappa} + \gamma \mathbf{u}) \right\} + \sum_{\mathbf{z} \in \boldsymbol{\Omega}_1}h(\mathbf{z}, \mathbf{f})\exp\left\{\mathbf{z}^T(\boldsymbol{\kappa} + \gamma \mathbf{u}) +\gamma \delta \right\}}{\sum_{\mathbf{b} \in \boldsymbol{\Omega}_0}\exp\left\{\mathbf{b}^T(\boldsymbol{\kappa} + \gamma \mathbf{u}) \right\} + \sum_{\mathbf{b} \in \boldsymbol{\Omega}_1}\exp\left\{\mathbf{b}^T(\boldsymbol{\kappa} + \gamma \mathbf{u}) + \gamma\delta \right\}}\\
&= \frac{A_0 + \exp(\delta\gamma)A_1}{D_0 + \exp(\delta\gamma)D_1}
\end{align*}
where $A_i, D_i$ do not depend on $\delta$ and the $D_i$ are all strictly positive.   From this point on the proof is identical to that of Proposition 19 in \citet{rosenbaum2002observational}: the partial derivative with respect to $\delta$ (computed below) has constant sign, which is sufficient for monotonicity.
\begin{align*}
\frac{\partial \omega(\mathbf{f}, \mathbf{u}+ \delta\mathbf{e}_{\ell i})}{ \partial \delta} = \frac{(D_0A_1 - A_0D_1)\gamma_\ell\exp(\gamma_\ell\delta)}{\left[D_0 + D_1\exp(\gamma_\ell\delta)\right]^2}
\end{align*}
\end{proof}

\subsection{Alternative large-sample null distribution accounting for propensity score estimation}
\label{subsec:m_est}

 In this section we construct and give intuition for an alternative variance estimator for the difference-in-means statistic under covariate-adaptive randomization inference that accounts for estimation of the propensity score.  To accomplish this we represent both  propensity score estimation and calculation of the final test statistic as tasks within a common M-estimation framework.  

 
 While M-estimators are typically used in settings where independent samples from an infinite superpopulation are available and the goal is to estimate a parameter of that infinite population \citep{stefanski2002calculus}, we focus instead on a setting with an infinite sequence of ever-larger finite populations, in each of which only the treatment random variable $\mathbf{Z}$ is random.  In addition, some subjects within each finite population are grouped into matched sets, and treatment indicators for subjects within the same matched set are highly dependent; on the other hand, propensity score estimation typically leverages all unmatched individuals as well as those in matched sets, and it is more reasonable to think of unmatched individuals' treatments as mutually independent.  As such, we adopt new notation as follows.

First, for a given finite population, we let $\mathbf{O}_{k} = \{\mathbf{Y}_{k}, \mathbf{Z}_k, \mathbf{X}_k\}$ represent the set of observed data associated with all the units in matched sets $k = 1, \ldots K$. 
 We  also introduce quantities $\mathbf{O}_{K+1}, \ldots, \mathbf{O}_{K'}$ where each $\mathbf{O}_k$ with $k > K$ gives the $Y,Z, X$ information associated with a single unmatched unit, so $n_k = 1$ for these units.  Letting $m = \sum^{K'}_{k=1}k$, We define $\mathbf{Y}$,  $\mathbf{Z}$, and $\mathbf{Y}(z)$ as M-tuples collecting values from all $K'$ sets, with $\mathbf{X}$ defined as the analogous $m \times p$ matrix.  We let  $\mathbf{Z}_k$ be independent of $\mathbf{Z}_{k'}$ for all $1 \leq k, k' \leq K', k \neq k'$, conditional on $\mathbf{Y}(1), \mathbf{Y}(0), \mathbf{X}$.
 
We assume that the propensity score is fit using M-estimation (such as a logistic regression).
 Let $\theta \in \mathbb{R}^p$ be the parameter vector for the propensity score model, and let 
 
 \[
 {\boldsymbol \psi} (\mathbf{O}_k, \theta)= (\psi_1(\mathbf{O}_k, \theta), \ldots, \psi_p(\mathbf{O}_k, \theta))^T
 \]
  be the set of associated functions that define the M-estimator via the estimating equations:
 \begin{align*}
\sum^{K'}_{k=1} {\boldsymbol \psi}(\mathbf{O}_k, \theta) = \mathbf{0}.
 \end{align*}
Note that for convenience we assume matched sets are given here and will not be influenced by the estimated propensity score (see Section 2.3 of the main manuscript for more discussion of such assumptions and their implications). 
 We augment these equations with three more components, one giving the proportion $\eta$ of sets from 1 to $K'$ that are included in the match, one giving the mean  $\mu$ of the difference-in-means statistic, and one giving the difference-in-means statistic $T$ itself.
 \begin{align*}
 \sum^{K'}_{k=1}\left(\begin{array}{c} 
 {\boldsymbol \psi}(\mathbf{O}_k, \theta)  \\
 1/\eta - 1\{n_k > 1\} \\
\eta  1\{n_k > 1\} \left\{ \mu - \sum^{n_k}_{i=1}Y_{ki}\left(\frac{n_k \cdot \text{odds}_{ki}(\theta)}{(n_k-1)\sum^{n_k}_{j=1}\text{odds}_{kj}(\theta)} - \frac{1}{n_k -1}\right)\right\} \\
 \eta  1\{n_k > 1\}\left\{T  - \sum^{n_k}_{i=1}Y_{ki}\frac{n_kZ_{ki} - 1}{n_k - 1} + \mu \right\}
 \end{array} \right) = \mathbf{0}
 \end{align*}
 We let ${\boldsymbol \psi}^{full}$ represent this expanded set of estimating equations.   
We may obtain parameter estimates $(\widehat{\theta}, \widehat{\mu}, T)$ by solving this equation (where $T$ is simply the actual difference-in-means statistic rather than an estimate of it). 

Assume that $\theta_0, \mu_0, \eta_0$ and $T_0$  exist as unique solutions to the following equation:
\[
\lim_{K' \longrightarrow \infty}\frac{1}{K'}\sum^{K'}_{k=1}E\left[{\boldsymbol \psi}^{full}(\mathbf{O}_k, \theta, \eta, \mu, T) \mid \mathbf{Y}(1), \mathbf{Y}(0), \mathbf{X}\right].
\]
In short, $\theta_0$, $\mu_0$, and $\eta_0$ are the large-sample limits of the propensity score parameter, the mean of the difference-in-means statistic covariate-adaptive randomization distribution, the ratio $K'/K$, and the test statistic itself across our sequence of growing finite populations  under the sharp null hypothesis.  Note that under these conditions we know $T_0 = 0$.

We now wish to estimate the asymptotic variance of $T$ under the sharp null hypothesis of no effect of treatment, accounting for the estimation of the propensity score.  We construct an estimator using the sandwich variance approach described for traditional infinite-superpopulation M-estimators in \citet{stefanski2002calculus}.  Specifically, we construct an estimator for one scalar entry of the following matrix, analogous to the (stabilized) asymptotic variance-covariance matrix of the parameter vector in the traditional infinite superpopulation version of M-estimation:
\[
V(\theta_0, \eta_0, \mu_0) = 
A(\theta_0, \eta_0, \mu_0, 0)^{-1}B(\theta_0,  \eta_0, \mu_0, 0)[A(\theta_0,  \eta_0, \mu_0, 0)^{-1}]^T
\]
 where $A(\cdot)$ and $B(\cdot)$ are defined as follows: 

\begin{align*}
A(\theta_0, \eta_0, \mu_0,0) &= 
\lim_{K' \longrightarrow \infty}\frac{1}{K'}\sum^{K'}_{k=1}E[-\nabla_{\theta, \eta, \mu,T}, {\boldsymbol \psi}^{full}(\mathbf{O}_k, \theta_0,  \eta_0, \mu_0, 0)] \\
B(\theta_0, \mu_0,0) &= \lim_{K' \longrightarrow \infty}\frac{1}{K'}\sum^{K'}_{k=1}E[{\boldsymbol \psi}^{full}(\mathbf{O}_k, \theta_0,  \eta_0, \mu_0, 0){\boldsymbol \psi}^{full}(\mathbf{O}_k, \theta_0,  \eta_0, \mu_0, 0)^T]
\end{align*}

Note that the formulas differ from the traditional ones in \citet{stefanski2002calculus}, which take expectation over a sampling distribution for all  observed elements $O_i$, by taking the limit of a sample average over the non-random elements $(Y_{ki})$; we assume here that those limits exist and that the limit $A(\theta_0, \eta_0, \mu_0)$ is invertible. We acknowledge that our use of this formula in the finite population is heuristic in the sense that we do not provide formal regularity conditions or a result guaranteeing 
that the asymptotic variance of $T$ is equal to $V(\theta)$ within our finite population framework, since this exceeds the scope of what we can accomplish in the supplement.

We now give more explicit forms for $A(\cdot)$ and $B(\cdot)$ to provide motivation for our specific estimator of ${V}(\theta_0, \eta_0, \mu_0)$ , using the fact that $E\psi_j(\theta_0,  \eta_0, \mu_0, 0) = 0$ to simplify. 

\small
\begin{align*}
  &A(\theta_0, \eta_0 \mu_0, 0) =\\ 
 &\lim_{K' \longrightarrow \infty}\left(\begin{array}{cccc} -\frac{1}{K'}\sum^{K'}_{k=1}E\left[\nabla_{\theta}{\boldsymbol \psi}(\mathbf{O}_k, \theta_0)\right] & 0 & 0 & 0 \\
0 & \eta_0^{-2} & 0 & 0 \\
  - \frac{\eta_0}{K'}\sum^{K}_{k=1}\sum^{n_k}_{i=1}Y_{ki}\frac{n_k}{n_k-1}\left(\frac{\partial}{\partial \theta_{1}} \frac{ \text{odds}_{ki}(\theta_0)}{\sum^{n_k}_{j=1} \text{odds}_{kj}(\theta_0)}\quad \cdots \quad\frac{\partial}{\partial \theta_{p}} \frac{ \text{odds}_{ki}(\theta_0)}{\sum^{n_k}_{j=1} \text{odds}_{kj}(\theta_0)}\right)  &0
   &  -1
   & 0 \\
 \mathbf{0}_{1 \times p} & 0
 & -1 & -1
 \end{array} \right)
 \end{align*}
 
 \begin{align*}
B(\theta_0, \eta_0, \mu_0, 0)  &=  \lim_{K' \longrightarrow \infty}\frac{1}{K'}\sum^{K'}_{k=1}E\left( \begin{array}{cccc} B_{\theta\theta} & B_{\theta\eta} & B_{\theta \mu} & B_{\theta T}\\
B_{\theta\eta}^T & B_{\eta\eta} & B_{\eta \mu} & B_{\eta T}\\
B_{\theta\mu}^T & B_{\eta\mu}^T & B_{\mu \mu} & B_{\mu T}\\
B_{\theta T}^T & B_{\eta T}^T & B_{\mu T}^T & B_{TT}
\end{array}\right)
 \end{align*}
 where
 \begin{align*}
 B_{\theta \theta} &= \frac{1}{K'}\sum^{K'}_{k=1}E\left\{{\boldsymbol \psi}(\mathbf{O}_k, \theta_0){\boldsymbol \psi}(\mathbf{O}_k, \theta_0)^T 
 \right\}\\
  B_{\theta \eta} &= \frac{1}{K'}\sum^{K'}_{k=1}E\left\{{\boldsymbol \psi}(\mathbf{O}_k, \theta_0)\left(\eta_0^{-1} - 1\{n_k > 1\} \right)
 \right\} = -\frac{\eta_0}{K'}\sum^K_{k=1}E\left\{ { \boldsymbol \psi}(\mathbf{O}_k, \theta_0)\right\}\\
  B_{\theta \mu} &= \frac{1}{K'}\sum^{K'}_{k=1}E\left\{ {\boldsymbol \psi}(\mathbf{O}_k, \theta_0)\eta_01\{n_k > 1\}\left[ \mu_0 - \sum^{n_k}_{i=1}Y_{ki}\left(\frac{n_k \cdot \text{odds}_{ki}(\theta_0)}{(n_k-1)\sum^{n_k}_{j=1}\text{odds}_{kj}(\theta_0)} - \frac{1}{n_k -1}\right)\right] 
 \right\}\\
 &= \frac{\eta_0}{K'}\sum^K_{k=1}E\left\{ { \boldsymbol \psi}(\mathbf{O}_k, \theta_0)\left[\mu_0 - \sum^{n_k}_{i=1}Y_{ki}\left(\frac{n_k \cdot \text{odds}_{ki}(\theta_0)}{(n_k-1)\sum^{n_k}_{j=1}\text{odds}_{kj}(\theta_0)} - \frac{1}{n_k -1}\right) \right] \right\} 
\\
  B_{\theta  T} &= \frac{1}{K'}\sum^{K'}_{k=1}E\left\{{ \boldsymbol \psi}(\mathbf{O}_k, \theta_0)\eta_01\{n_k > 1\}\left[ T  - \sum^{n_k}_{i=1}Y_{ki}\frac{n_kZ_{ki} - 1}{n_k - 1} + \mu  \right] 
 \right\}\\
 &=\frac{\eta_0}{K'}\sum^K_{k=1}E\left\{{ \boldsymbol \psi}(\mathbf{O}_k, \theta_0)\left[\mu_0 - \sum^{n_k}_{i=1}Y_{ki}\frac{n_kZ_{ki} - 1}{n_k - 1}  \right]\right\} \\
  B_{\eta \eta} &= \frac{1}{K'}\sum^{K'}_{k=1}E\left\{ \left(\eta_0^{-1} - 1\{n_k > 1\} \right)
 \right\} = \eta_0^{-2} + \frac{K}{K'}(1 -2\eta_0^{-1}) \\
  B_{\eta \mu} &= \frac{1}{K'}\sum^{K'}_{k=1}E\left\{1\{n_k > 1\}\left(1 - \eta_0 \right)\left[ \mu - \sum^{n_k}_{i=1}Y_{ki}\left(\frac{n_k \cdot \text{odds}_{ki}(\theta_0)}{(n_k-1)\sum^{n_k}_{j=1}\text{odds}_{kj}(\theta_0)} - \frac{1}{n_k -1}\right)\right]
 \right\}\\
 &=  \frac{1 - \eta_0}{K'} \sum^K_{k=1}\left[ \mu_0 - \sum^{n_k}_{i=1}Y_{ki}\left(\frac{n_k \cdot \text{odds}_{ki}(\theta_0)}{(n_k-1)\sum^{n_k}_{j=1}\text{odds}_{kj}(\theta_0)} - \frac{1}{n_k -1}\right)\right] \\
  B_{\eta T} &= \frac{1}{K'}\sum^{K'}_{k=1}E\left\{1\{n_k > 1\}\left(1 - \eta_0 \right) \left[ T  - \sum^{n_k}_{i=1}Y_{ki}\frac{n_kZ_{ki} - 1}{n_k - 1} + \mu  \right] 
 \right\}\\
 &=  \frac{1 - \eta_0}{K'} \sum^K_{k=1} \left[ \mu_0  - \sum^{n_k}_{i=1}Y_{ki}\left(\frac{n_k \cdot \text{odds}_{ki}(\theta_0)}{(n_k-1)\sum^{n_k}_{j=1}\text{odds}_{kj}(\theta_0)} - \frac{1}{n_k -1}\right)  \right] \\
  B_{\mu \mu} &= \frac{1}{K'}\sum^{K'}_{k=1}E\left\{\eta_0^21\{n_k > 1\}\left[ \mu - \sum^{n_k}_{i=1}Y_{ki}\left(\frac{n_k \cdot \text{odds}_{ki}(\theta_0)}{(n_k-1)\sum^{n_k}_{j=1}\text{odds}_{kj}(\theta_0)} - \frac{1}{n_k -1}\right)\right]^2 
 \right\}\\
 &=  \frac{ \eta_0^2}{K'} \sum^K_{k=1}\left[ \mu_0 - \sum^{n_k}_{i=1}Y_{ki}\left(\frac{n_k \cdot \text{odds}_{ki}(\theta_0)}{(n_k-1)\sum^{n_k}_{j=1}\text{odds}_{kj}(\theta_0)} - \frac{1}{n_k -1}\right)\right]^2 \\
  B_{\mu T} &= \frac{1}{K'}\sum^{K'}_{k=1}E\left\{eta_0^21\{n_k > 1\}\left[ \mu - \sum^{n_k}_{i=1}Y_{ki}\left(\frac{n_k \cdot \text{odds}_{ki}(\theta_0)}{(n_k-1)\sum^{n_k}_{j=1}\text{odds}_{kj}(\theta_0)} - \frac{1}{n_k -1}\right)\right] \left[ T  - \sum^{n_k}_{i=1}Y_{ki}\frac{n_kZ_{ki} - 1}{n_k - 1} + \mu  \right]
 \right\}\\
 &= \frac{ \eta_0^2}{K'} \sum^K_{k=1}\left[ \mu_0 - \sum^{n_k}_{i=1}Y_{ki}\left(\frac{n_k \cdot \text{odds}_{ki}(\theta_0)}{(n_k-1)\sum^{n_k}_{j=1}\text{odds}_{kj}(\theta_0)} - \frac{1}{n_k -1}\right)\right]^2\\
  B_{TT} &= \frac{1}{K'}\sum^{K'}_{k=1}E\left\{\eta_0^21\{n_k > 1\}\left[ T  - \sum^{n_k}_{i=1}Y_{ki}\frac{n_kZ_{ki} - 1}{n_k - 1} + \mu  \right]^2 
 \right\}\\
 &=  \frac{ \eta_0^2}{K'} \sum^K_{k=1}\left\{\mu_0^2 - 2\mu_0\ \sum^{n_k}_{i=1}Y_{ki}\left(\frac{n_k \cdot \text{odds}_{ki}(\theta_0)}{(n_k-1)\sum^{n_k}_{j=1}\text{odds}_{kj}(\theta_0)} - \frac{1}{n_k -1}\right) + E\left[ \sum^{n_k}_{i=1}Y_{ki}\frac{n_kZ_{ki} - 1}{n_k - 1} \right]^2\right\}
 \end{align*}
 

\normalsize



  For ease of presentation we let $A(\theta_0)$ and $B(\theta_0)$ be analogous to $A(\theta_0, \eta_0, \mu_0, 0)$ and  $B(\theta_0, \eta_0, \mu_0, 0)$  for the smaller M-estimation problem of fitting the propensity score alone, and  we define  $A(\theta_0, \eta_0)$, $B(\theta_0, \eta_0)$, $A(\theta_0, \eta_0, \mu_0)$ and $B(\theta_0, \eta_0, \mu_0)$ analogously. 
We now invert $A(\theta_0, \eta_0, \mu_0, 0)$.    We proceed by representing it as a block matrix:

\begin{align*}
A(\theta_0, \eta_0, \mu_0, 0)^{-1} &= \left(\begin{array}{cc} 
A(\theta_0, \eta_0,  \mu_0) & \mathbf{0} \\
(\mathbf{0},-1) & -1
\end{array}\right)^{-1}
=  \left(\begin{array}{cc} 
A(\theta_0, \eta_0,  \mu_0)^{-1} & \mathbf{0} \\
(\mathbf{0},-1)A(\theta_0, \eta_0,  \mu_0)^{-1} & -1
\end{array}\right)
\end{align*}
We use a similar strategy to compute $A(\theta_0, \eta_0,  \mu_0)^{-1}$:
\begin{align*}
A(\theta_0,  \eta_0,  \mu_0)^{-1}  &=\left(\begin{array}{cc} 
A(\theta_0,  \eta_0) & \mathbf{0} \\
\left(\lim_{K' \longrightarrow \infty}\frac{\eta_0}{K}\sum^K_{k=1}\sum^{n_k}_{i=1}Y_{ki}\frac{n_k}{n_k-1}\mathbf{d}_k(\theta_0), 0 \right) & -1
\end{array}\right)^{-1}\\
&=  \left(\begin{array}{cc} 
A(\theta_0,  \eta_0)^{-1} & \mathbf{0} \\
\left(\lim_{K' \longrightarrow \infty}\frac{\eta_0}{K}\sum^K_{k=1}\sum^{n_k}_{i=1}Y_{ki}\frac{n_k}{n_k-1}\bdk(\theta_0), 0 \right)A(\theta_0,  \eta_0)^{-1}  & -1
\end{array}\right)
\end{align*}
where $\mathbf{d}_k(\theta_0) = \left(\frac{\partial}{\partial \theta_{1}} \frac{  \text{odds}_{ki}(\theta_0)}{\sum^{n_k}_{j=1}\text{odds}_{kj}(\theta_0)}\quad \cdots \quad\frac{\partial}{\partial \theta_{p}} \frac{ \text{odds}_{ki}(\theta_0)}{\sum^{n_k}_{j=1}\text{odds}_{kj}(\theta_0)}\right)$ for each $k = 1, \ldots, K$.
Finally, we compute $A(\theta_0, \eta_0)^{-1}$:
\begin{align*}
A(\theta_0,  \eta_0)^{-1} &= \left(\begin{array}{cc} 
A(\theta_0) & \mathbf{0} \\
\mathbf{0} & \eta_0^{-2}
\end{array}\right)^{-1}
=  \left(\begin{array}{cc} 
A(\theta_0)^{-1} & \mathbf{0} \\
\mathbf{0} & \eta_0^{2}
\end{array}\right)
\end{align*}

We now evaluate our quantity of interest, entry $((p + 3), (p+3))$  of the matrix $V(\theta_0,\eta_0,\mu_0)$, which we will denote $v_T(\theta_0,\eta_0,\mu_0)$.
\begin{align*}
v_T(\theta_0,\eta_0,\mu_0) = \mathbf{a}(\theta_0)B(\theta_0, \eta_0, \mu_0, 0)\mathbf{a}(\theta_0)^T
\end{align*}
where 
\begin{align*}
\mathbf{a}(\theta_0) =
& \left(\begin{array}{cccc} -\lim_{K' \longrightarrow \infty}\frac{\eta_0}{K'}\sum^K_{k=1}\sum^{n_k}_{i=1}Y_{ki}\frac{n_k}{n_k-1}\bdk(\theta_0) A(\theta_0)^{-1} & 0  & 1 &-1\end{array}\right).
\end{align*}
Expanding this, we obtain the following:
\small
\begin{align*}
&
v_T(\theta_0,\eta_0,\mu_0) =\\
& \left\{\lim_{K' \longrightarrow \infty}\frac{1}{K'}\sum^K_{k=1}\sum^{n_k}_{i=1}Y_{ki}\frac{n_k}{n_k-1}\bdk(\theta_0)\right\}A(\theta_0)^{-1} B(\theta_0)[A(\theta_0)^{-1}]^T\left\{\lim_{K' \longrightarrow \infty}\frac{\eta_0}{K'}\sum^K_{k=1}\sum^{n_k}_{i=1}Y_{ki}\frac{n_k}{n_k-1}\bdk(\theta_0)\right\}^T\\
&-2 \left\{\lim_{K' \longrightarrow \infty}\frac{\eta_0}{K'}\sum^K_{k=1}\sum^{n_k}_{i=1}Y_{ki}\frac{n_k}{n_k-1}\bdk(\theta_0)\right\}A(\theta_0)^{-1}\Bigg\{\\\
&\hspace{10em}\left.\lim_{K' \longrightarrow \infty} \frac{\eta_0}{K'}\sum^K_{k=1}E\left\{ { \boldsymbol \psi}(\mathbf{O}_k, \theta_0)\left[\mu - \sum^{n_k}_{i=1}Y_{ki}\left(\frac{n_k \cdot \text{odds}_{ki}(\theta_0)}{(n_k-1)\sum^{n_k}_{j=1}\text{odds}_{kj}(\theta_0)} - \frac{1}{n_k -1}\right) \right]
\right\}  \right\}\\
& -2 \left\{\lim_{K' \longrightarrow \infty}\frac{\eta_0}{K'}\sum^K_{k=1}\sum^{n_k}_{i=1}Y_{ki}\frac{n_k}{n_k-1}\bdk(\theta_0)\right\}A(\theta_0)^{-1}
\left\{\lim_{K' \longrightarrow \infty}\frac{\eta_0}{K'}\sum^K_{k=1}E\left\{{ \boldsymbol \psi}(\mathbf{O}_k, \theta_0)\left[ \mu  - \sum^{n_k}_{i=1}Y_{ki}\frac{n_kZ_{ki} - 1}{n_k - 1}   \right]\right\}  \right\}\\
& +  \lim_{K' \longrightarrow \infty}\frac{\eta_0^2}{K'}\sum^K_{k=1}E\left[ \sum^{n_k}_{i=1}Y_{ki}\frac{n_kZ_{ki} - 1}{n_k - 1} \right]^2- \lim_{K' \longrightarrow \infty}\frac{\eta_0^2}{K'}\sum^K_{k=1} \left[\sum^{n_k}_{i=1}Y_{ki}\left(\frac{n_k \cdot \text{odds}_{ki}(\theta_0)}{(n_k-1)\sum^{n_k}_{j=1}\text{odds}_{kj}(\theta_0)} - \frac{1}{n_k -1}\right)\right]^2 \\
&= \left\{\lim_{K' \longrightarrow \infty}\frac{\eta_0}{K'}\sum^K_{k=1}\sum^{n_k}_{i=1}Y_{ki}\frac{n_k}{n_k-1}\bdk(\theta_0)\right\}A(\theta_0)^{-1} B(\theta_0)[A(\theta_0)^{-1}]^T\left\{\lim_{K' \longrightarrow \infty}\frac{\eta_0}{K'}\sum^K_{k=1}\sum^{n_k}_{i=1}Y_{ki}\frac{n_k}{n_k-1}\bdk(\theta_0)\right\}^T\\
& -2 \left\{\lim_{K' \longrightarrow \infty}\frac{\eta_0}{K'}\sum^K_{k=1}\sum^{n_k}_{i=1}Y_{ki}\frac{n_k}{n_k-1}\bdk(\theta_0)\right\}A(\theta_0)^{-1}
\Bigg\{\\\
&\hspace{10em}\left.\lim_{K' \longrightarrow \infty} \frac{\eta_0}{K'}\sum^K_{k=1}E\left\{ { \boldsymbol \psi}(\mathbf{O}_k, \theta_0)\left[\mu - \sum^{n_k}_{i=1}Y_{ki}\left(\frac{n_k \cdot \text{odds}_{ki}(\theta_0)}{(n_k-1)\sum^{n_k}_{j=1}\text{odds}_{kj}(\theta_0)} - \frac{1}{n_k -1}\right) \right] 
\right\}  \right\}\\
& -2 \left\{\lim_{K' \longrightarrow \infty}\frac{\eta_0}{K'}\sum^K_{k=1}\sum^{n_k}_{i=1}Y_{ki}\frac{n_k}{n_k-1}\bdk(\theta_0)\right\}A(\theta_0)^{-1}
\left\{\lim_{K' \longrightarrow \infty}\frac{\eta_0}{K'}\sum^K_{k=1}E\left\{{ \boldsymbol \psi}(\mathbf{O}_k, \theta_0)\left[ \mu  - \sum^{n_k}_{i=1}Y_{ki}\frac{n_kZ_{ki} - 1}{n_k - 1}   \right]\right\}  \right\}\\
& + \lim_{K' \longrightarrow \infty}\frac{\eta_0^2}{K'}\sum^K_{k=1}\sum^{n_k}_{i=1}\left(\frac{n_k}{n_k-1}\right)^2Y_{ki}p_{ki}\left\{Y_{ki}(1- p_{ki}) - \sum^{n_k}_{j\neq i}Y_{kj}p_{kj}\right\} 
\end{align*}

\normalsize

Notice that the last term is identical to the variance of $T$ when the true propensity scores are known, derived above (after normalizing both sides by $K$). 
While $v_T(\theta_0,\eta_0,\mu_0)$ cannot be calculated in practice without knowledge of the true parameters, the moments of the derivatives of the ${\boldsymbol \psi}$ functions, and the limits in $K$, we can estimate it by substituting estimated parameters and sample moments/estimates as follows. For convenience, we let $\mathbf{O}_{k}^\ell = \{\mathbf{Y}_{k}, \mathbf{Z}_{k}^\ell, \mathbf{X}_k\}$ where $\mathbf{Z}_{k}^\ell$ be the length-$n_k$ vector containing a 1 at element $\ell$ and zeroes otherwise.  Note also that to estimate the variance of $T$, rather than the variance of $\sqrt{K'}(T-\mu)$, we must also divide our estimate for $v_T(\theta_0,\eta_0,\mu_0)$ by $K'$.

\footnotesize
\begin{align*}
&\widehat{Var}(T) = \frac{1}{K'}\widehat{v}_T(\widehat{\theta}, \widehat{\eta}, \widehat{\mu}) = \\
& \left\{\frac{1}{K}\sum^{K}_{k=1}\sum^{n_k}_{i=1}Y_{ki}\frac{n_k}{n_k-1}\bdk(\widehat{\theta}_0)\right\}\frac{A_n(\widehat{\theta}_0)^{-1} B_n(\widehat{\theta}_0)[A_n(\widehat{\theta}_0)^{-1}]^T}{K'}\left\{\frac{1}{K}\sum^{K}_{k=1}\sum^{n_k}_{i=1}Y_{ki}\frac{n_k}{n_k-1}\bdk(\widehat{\theta}_0)\right\}^T\\
& -\frac{2}{K'} \left\{\frac{1}{K}\sum^K_{k=1}\sum^{n_k}_{i=1}Y_{ki}\frac{n_k}{n_k-1}\bdk(\widehat{\theta}_0)\right\}A_n(\widehat{\theta}_0)^{-1}
\left\{ \frac{1}{K}\sum^K_{k=1}\sum^{n_k}_{\ell =1}\frac{\text{odds}_{k\ell}(\widehat{\theta}_0)}{\sum^{n_k}_{j=1}\text{odds}_{kj}(\widehat{\theta}_0)}{ \boldsymbol \psi}(\mathbf{O}_k^\ell, \widehat{\theta_0})\Bigg[ \right.\\
&\hspace{4em}\left.\left.\frac{1}{K}\sum^K_{k'=1}\sum^{n_{k'}}_{i=1}Y_{ki}\left(\frac{n_{k'} \cdot \text{odds}_{k'i}(\theta)}{(n_{k'}-1)\sum^{n_{k'}}_{j=1}\text{odds}_{k'j}(\theta)} - \frac{1}{n_{k'} -1}\right) - \sum^{n_k}_{i=1}Y_{ki}\left(\frac{n_k \cdot \text{odds}_{ki}(\widehat{\theta_0})}{(n_k-1)\sum^{n_k}_{j=1}\text{odds}_{kj}(\widehat{\theta_0})} - \frac{1}{n_k -1}\right) \right] 
 \right\}\\
& -\frac{2}{K'} \left\{\frac{1}{K}\sum^K_{k=1}\sum^{n_k}_{i=1}Y_{ki}\frac{n_k}{n_k-1}\bdk(\widehat{\theta}_0)\right\}A_n(\widehat{\theta}_0)^{-1}
\left\{\frac{1}{K}\sum^K_{k=1}\sum^{n_k}_{\ell =1}\frac{\text{odds}_{k\ell}(\widehat{\theta}_0)}{\sum^{n_k}_{j=1}\text{odds}_{kj}(\widehat{\theta}_0)}{ \boldsymbol \psi}(\mathbf{O}_k^\ell, \widehat{\theta}_0)\Bigg[ \right.\\
&\hspace{17em}\left.\left. \frac{1}{K}\sum^K_{k'=1}\sum^{n_{k'}}_{i=1}Y_{ki}\left(\frac{n_{k'} \cdot \text{odds}_{k'i}(\theta)}{(n_{k'}-1)\sum^{n_{k'}}_{j=1}\text{odds}_{k'j}(\theta)} - \frac{1}{n_{k'} -1}\right)   - \sum^{n_k}_{i=1}Y_{ki}\frac{n_kZ^\ell_{ki} - 1}{n_k - 1}   \right] \right\}\\
&+ \frac{1}{K^2}\sum^K_{k=1}\sum^{n_k}_{i=1}\left(\frac{n_k}{n_k-1}\right)^2Y_{ki}\left(\frac{\text{odds}_{ki}(\widehat{\theta}_0)}{\sum^{n_k}_{\ell=1}\text{odds}_{k\ell}(\widehat{\theta}_0)}\right)\left\{Y_{ki}\left[1- \left(\frac{\text{odds}_{ki}(\widehat{\theta}_0)}{\sum^{n_k}_{\ell=1}\text{odds}_{k\ell}(\widehat{\theta}_0)}\right)\right] - \sum^{n_k}_{j\neq i}Y_{kj}\left(\frac{\text{odds}_{kj}(\widehat{\theta}_0)}{\sum^{n_k}_{\ell=1}\text{odds}_{k\ell}(\widehat{\theta}_0)}\right)\right\} 
\end{align*}
\normalsize
We now give a closed form for the $\ell$th element of $\bdk(\theta_0)$, i.e. 
$\frac{\partial}{\partial \theta_{\ell}}
 \frac{ \text{odds}_{ki}({\theta}_0)}{\sum^{n_k}_{j=1}\text{odds}_{kj}({\theta}_0)}$, 
 in the special case of the logistic regression model defined by $\log(P(Z=1|X)/P(Z=0|X)) = X\theta$ where the first column of $X$ consists entirely of 1s (and where we index the $\ell$th column of $X_{ki}$ as $X_{\ell ki}$).
\begin{align*}
\frac{ \text{odds}_{ki}({\theta})}{\sum^{n_k}_{j=1}\text{odds}_{kj}({\theta})} &= \frac{\exp(X_{ki}\theta)}{\sum^{n_k}_{j=1}\exp(X_{kj}\theta)} = \frac{1}{\sum^{n_k}_{j=1}\exp[(X_{kj}-X_{ki})\theta]} \\
\frac{\partial}{\partial \theta_{\ell}}\frac{ \text{odds}_{ki}({\theta})}{\sum^{n_k}_{j=1}\text{odds}_{kj}({\theta})} &= -\frac{1}{\left\{\sum^{n_k}_{j=1}\exp[(X_{kj}-X_{ki})\theta]\right\}^2}\sum^{n_k}_{j=1}\exp[(X_{kj}-X_{ki})\theta](X_{\ell kj} -  X_{\ell k i})
\end{align*}
This quantity will be zero for $\ell = 1$ since $X_{1 k j} = X_{1 k i} = 1$ for all $i,j$.
The estimating equations  ${ \boldsymbol \psi}(\mathbf{O}_k, \widehat{\theta}_0)$ are discussed in many standard textbooks and have the following form in our notation:
\[
\psi_j(\mathbf{O}_k, \widehat{\theta}_0) = \frac{1}{K'}\sum^{K'}_{k=1}\sum^{n_k}_{i=1}\left(Z_{ki} - \frac{1}{1 + \exp(-X_{ki}\theta)}\right)X_{j,ki}
\]
where $X_{1,ki} = 1$ for all $k,i$.  In turn, the gradient of $\psi_j$ has the following form:
\[
\nabla \psi_j(\mathbf{O}_k, \widehat{\theta}_0) =  -\sum^{n_k}_{i=1}\left( \begin{array}{ccccc} \frac{\exp(-X_{ki}\theta)X_{1,ki}X_{j,ki}}{[1 + \exp(-X_{ki}\theta)]^2}, & \cdots &, \frac{\exp(-X_{ki}\theta)X^2_{j,ki}}{[1 + \exp(-X_{ki}\theta)]^2},  & \cdots & ,\frac{\exp(-X_{ki}\theta)X_{p,ki}X_{j,ki}}{[1 + \exp(-X_{ki}\theta)]^2} 
\end{array}\right)
\]
Sample averages of these gradients (multiplied by $-1$ and evaluated at estimated values of $\theta_0$) form the rows of the matrix $A_n(\widehat{\theta}_0)$ in the variance estimate given above.  Note also that the term $A_n(\widehat{\theta}_0)^{-1} B_n(\widehat{\theta}_0)[A_n(\widehat{\theta}_0)^{-1}]^T$ is simply an estimate of the variance-covariance matrix for propensity score estimation; in our framework, in which finite populations exhibit internal matched-set structure and in which individual $O_k$ terms may correspond to clusters of units rather than distinct subjects, this is a clustered covariance estimate such as that of  \citet{liang1986longitudinal} with clustering over matched sets (and unmatched individuals treated as their own clusters).

The above all assumes that the propensity score is estimated in the same sample in which the match is conducted.  Note that Theorems 1-2 instead assume the propensity score estimation is done in a different sample.  In this case the variance estimation problem becomes simpler and does not actually require an M-estimation framework, since the propensity score estimates and the $Z_{ki}$s are independent.  Instead, the delta method can be used to calculate the variance of $\widehat{\mu}$ and this variance can be added to the estimated variance of $T$ to obtain an overall variance.  More specifically,  we can write:
\begin{align*}
h(\theta) = \sum^{n_k}_{i=1}Y_{ki}\left(\frac{n_k \cdot \text{odds}_{ki}(\theta)}{(n_k-1)\sum^{n_k}_{j=1}\text{odds}_{kj}(\theta)} - \frac{1}{n_k -1}\right)
\end{align*}
Let $\Sigma$ be the asymptotic variance of $\widehat{\theta}$.  Then by the delta method, 
\begin{align*}
\text{Var}(h(\theta)) \approx \nabla h(\theta)^T\Sigma \nabla h(\theta).
\end{align*}
where 
\begin{align*}
\nabla h(\theta) = \frac{1}{K}\sum^K_{k=1}\sum^{n_k}_{i=1}Y_{ki}\left(\frac{n_k}{n_k-1}\right)\left(\frac{\partial}{\partial \theta_{1}} \frac{ \text{odds}_{ki}(\theta)}{\sum^{n_k}_{j=1}\text{odds}_{kj}(\theta)}\quad \cdots \quad\frac{\partial}{\partial \theta_{p}} \frac{ \text{odds}_{ki}(\theta)}{\sum^{n_k}_{j=1} \text{odds}_{kj}(\theta)}\right)
\end{align*}
This quantity may be estimated by substituting $\widehat{\theta}$ for $\theta$.  Note that the resulting variance estimate is almost identical to $\widehat{\text{Var}}(T)$ as given above for the case of in-sample estimation, except with no cross-terms and with matrix $\Sigma$, which must be estimated in the pilot sample, in place of the in-sample covariance matrix for $\widehat{\theta}$.

\subsection{Additional simulation results}
\label{subsec:addtl_sims}

We provide further detail on the simulation settings explored in Section 6 
of the main manuscript and give results for several alternate settings.  First, for the primary setting considered in the main manuscript, Figure \ref{fig:sim_densities} shows smoothed density plots comparing the true propensity score distribution across treatment and control groups for sample draws from each of six different simulation settings affecting the treatment model.  The density plots are scaled by the relative size of the two groups; in large samples, regions of the plots where the control density exceeds the treated density indicate propensity score values at which near-exact matching on the propensity score is possible, while regions with larger treated density suggest regions in which either poorer matches must be accepted or treated units must be trimmed \citep{li2013weighting}; the number given in the title of the plot is an estimate of the proportion of the probability mass in the treated group that overlaps with the scaled control density, indicating easy matchability.  In all the settings shown here, most treated units can be matched closely but in most cases there is some region in which matching is more difficult, ensuring that inexact matches occur systematically.  Of course these plots are smoothed depictions that mask the underlying discreteness of the data, particularly in the $n=100$ case, but they suggest the general difficulty of the matching problem.

The remaining figures repeat the analyses given in Figure 1 from the main manuscript
and Figure \ref{fig:sim_densities} for several other simulation settings as discussed in Section 6. 
In particular, Figures \ref{fig:ps_dens} - \ref{fig:ps_signal} give results for a weak propensity score setting where the treated and control groups are more similar than in the primary simulations; Figure \ref{fig:rejection_heatmaps} gives results averaging only over simulations in which good covariate balance was achieved in the sense that absolute standardized differences in means for all measured covariates were under 0.2 after matching; and Figure \ref{fig:oos_heatmaps} gives results for a setting in which propensity scores were estimated in a large external sample.  The weak-signal results show similar patterns to the primary simulations but the size of type I error violations is substantially reduced and uniform randomization inference sometimes achieves Type I error control.  The density plots reveal that the weak-signal settings correspond to a setting in which the distribution of true propensity scores is very similar between groups; when propensity score distributions show such similarity in real datasets this suggests that uniform randomization inference may be warranted. The other two simulation settings produce patterns of Type I error rates almost identical to those in Figure 1
in the main manuscript.  



\begin{figure}
\includegraphics{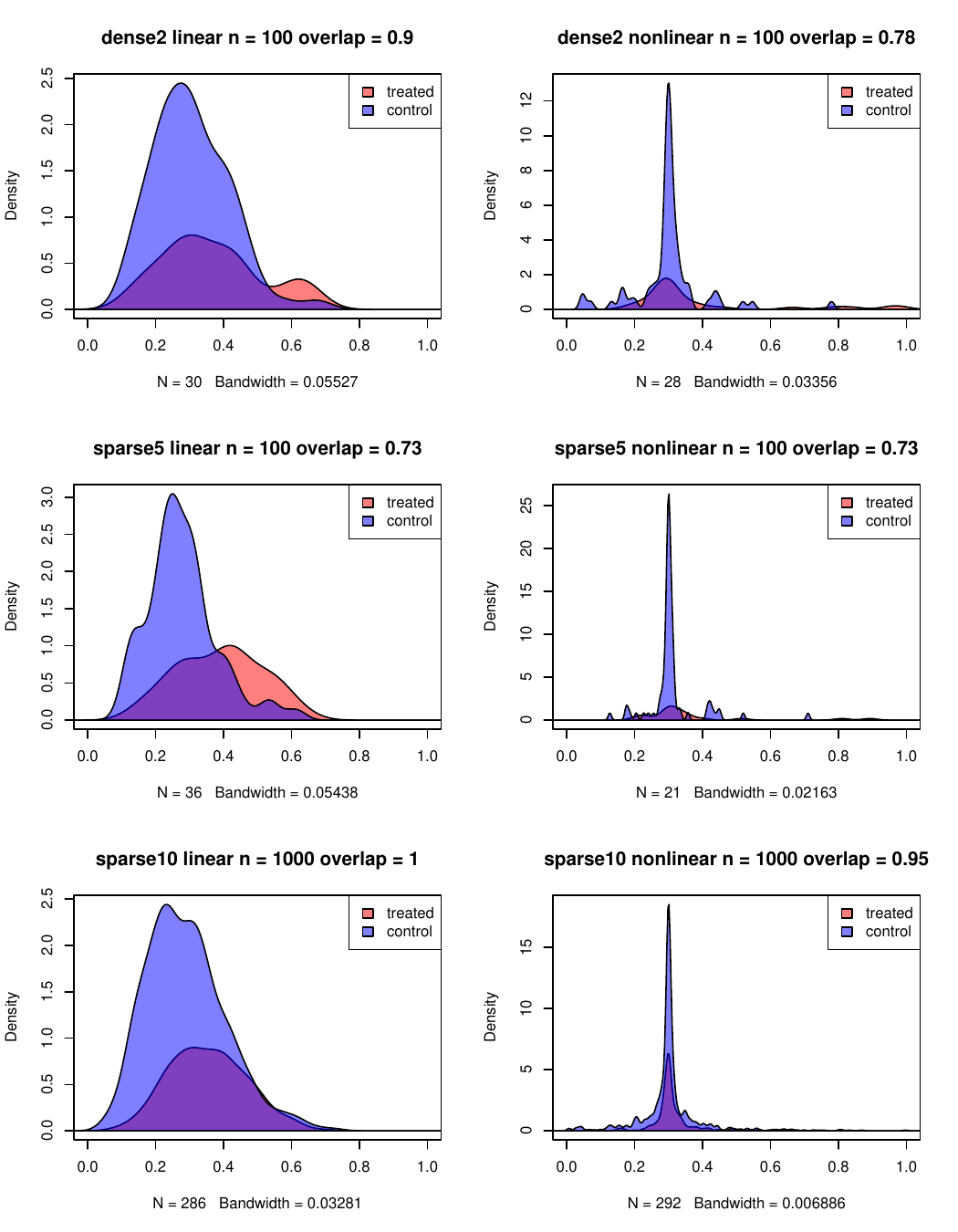}
\caption{\small Smoothed and scaled density plots of one random draw each from six different simulation settings for the treatment model. The first and second columns contrast weak and strong propensity score signals ($\tau = 0.2$ vs. $\tau = 0.6$), and the three rows contrast three different sizes for the initial dataset: $(n,p) = (100,2), (100,5)$, and (1000,10) respectively.}
\label{fig:sim_densities}
\end{figure}

\begin{figure}
\centering
\includegraphics{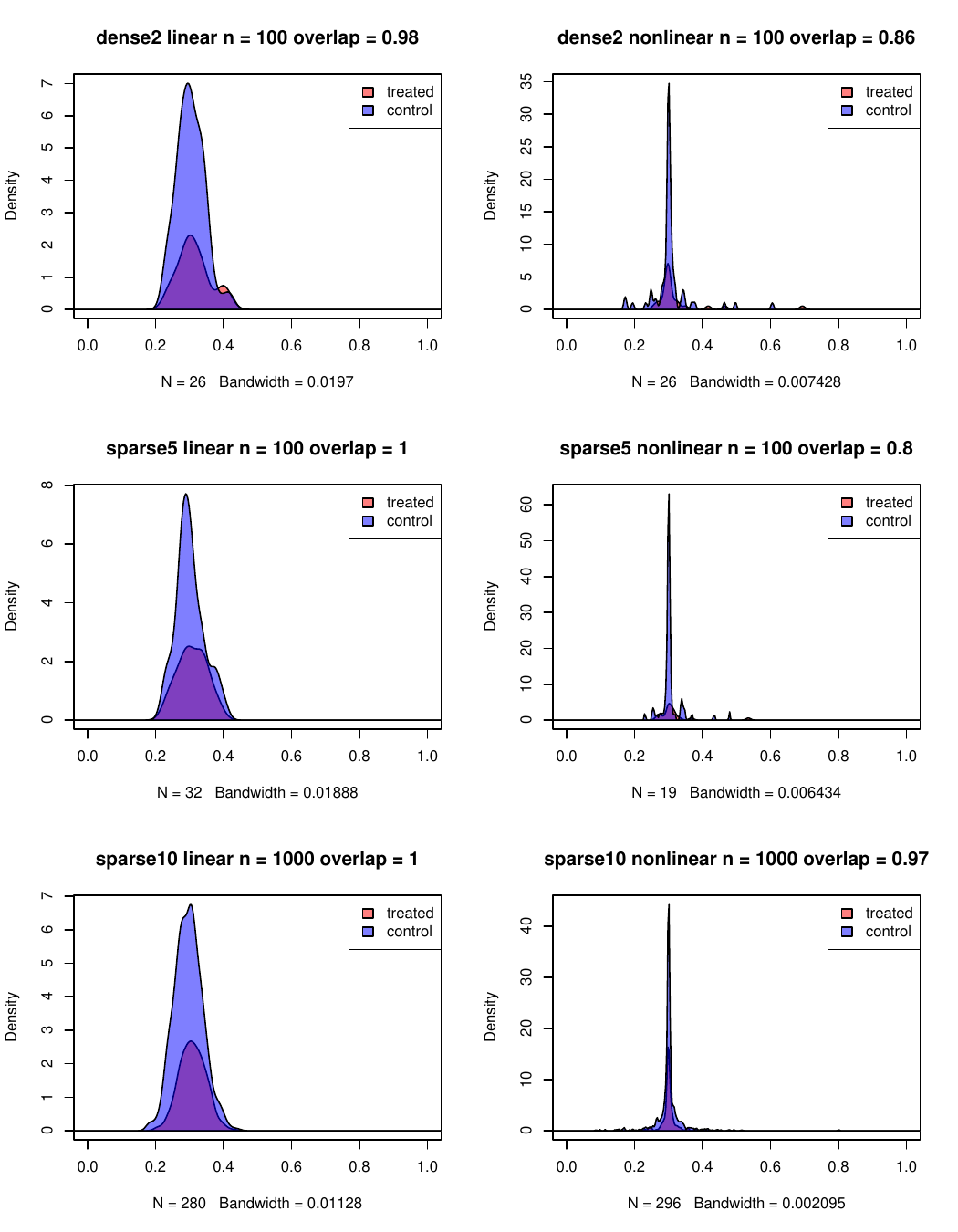}
\caption{\small Relative density of treated and control groups by propensity score value under a weak propensity score signal.  For more detail, see caption to Figure \ref{fig:sim_densities}.}
\label{fig:ps_dens}
\end{figure}

\begin{figure}
\centering
\includegraphics{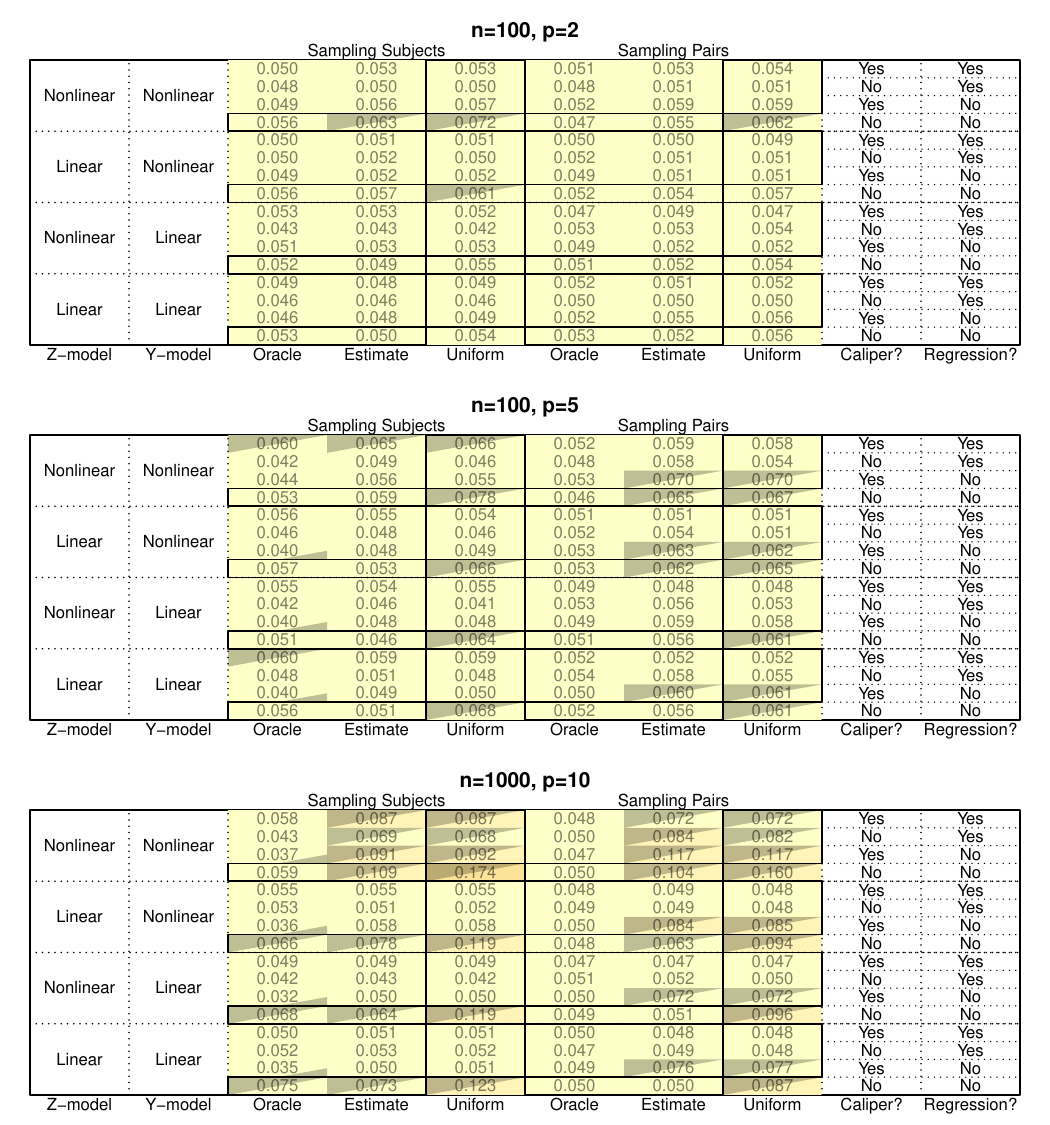}
\caption{\small Type I error results under weak propensity score signal for uniform and covariate-adaptive inference across multiple simulation settings. Each of the three tables corresponds to a separate dataset size.  Within each table the first three columns  contrast three inferential approaches when the subjects are sampled independently prior to matching; the last three columns do the same comparison when treatment assignments within matched sets are sampled independently after matching.  The rows of the table demonstrate different combinations of calipers and regression adjustment and correct or incorrect specification of treatment and outcome models.  Numbers give type I error rates with colors associated to their magnitude; triangles indicate that a one-sample z-test rejected the null hypothesis that the error rate was 0.05 (under a Bonferroni correction scaled to the number of results across the entire figure), with a large upper triangle indicating a positive z-statistic and a small lower triangle indicating a negative z-statistic.} 
\label{fig:ps_signal}
\end{figure}

%

\begin{figure}
\centering
\includegraphics{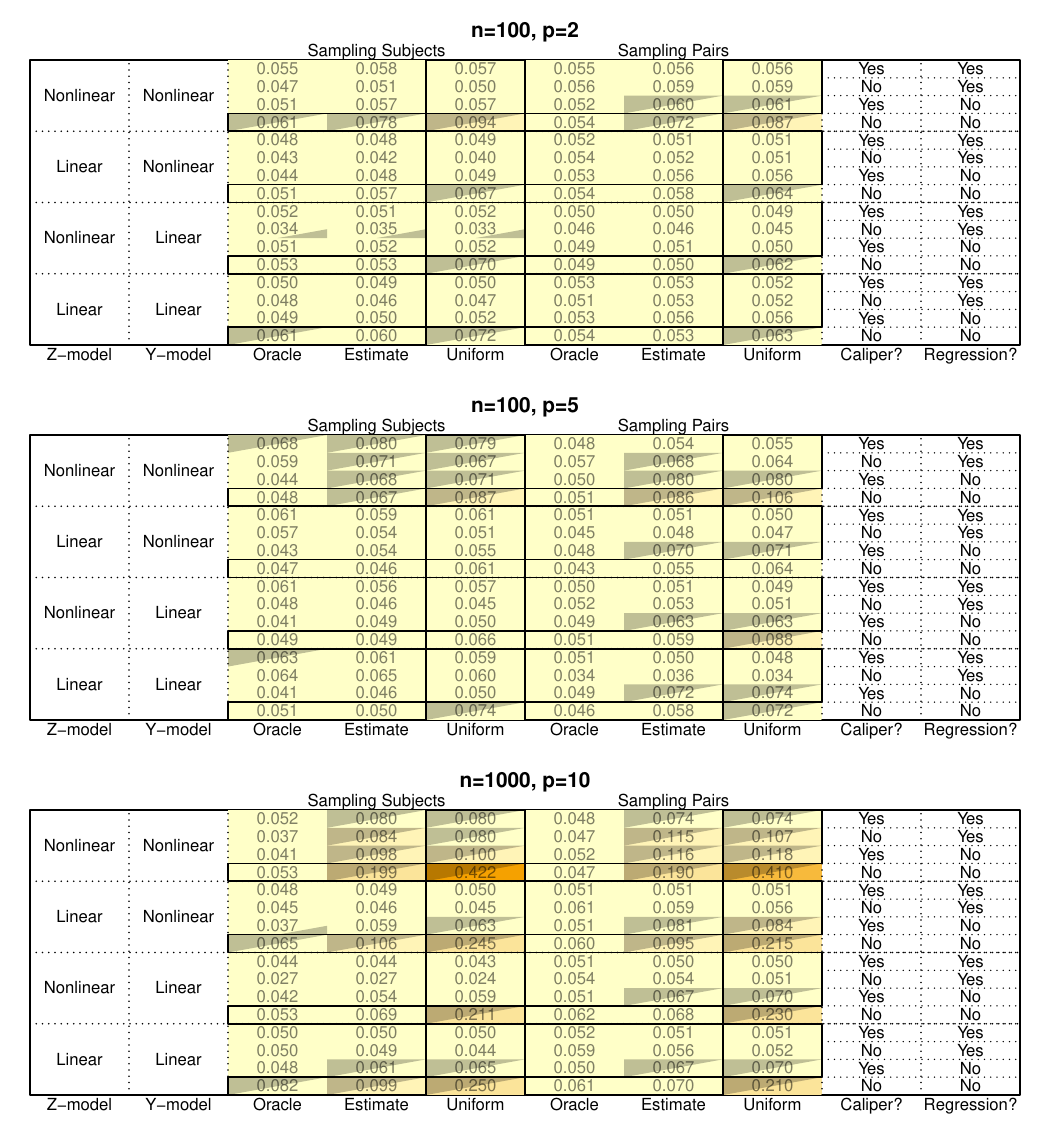}
\caption{\small Type I error results conditional on good covariate balance for uniform and covariate-adaptive inference across multiple simulation settings.  For a more thorough description of how the tables are organized see the caption to Figure \ref{fig:ps_signal}.}\label{fig:rejection_heatmaps}
\end{figure}

\begin{figure}
\centering
\includegraphics{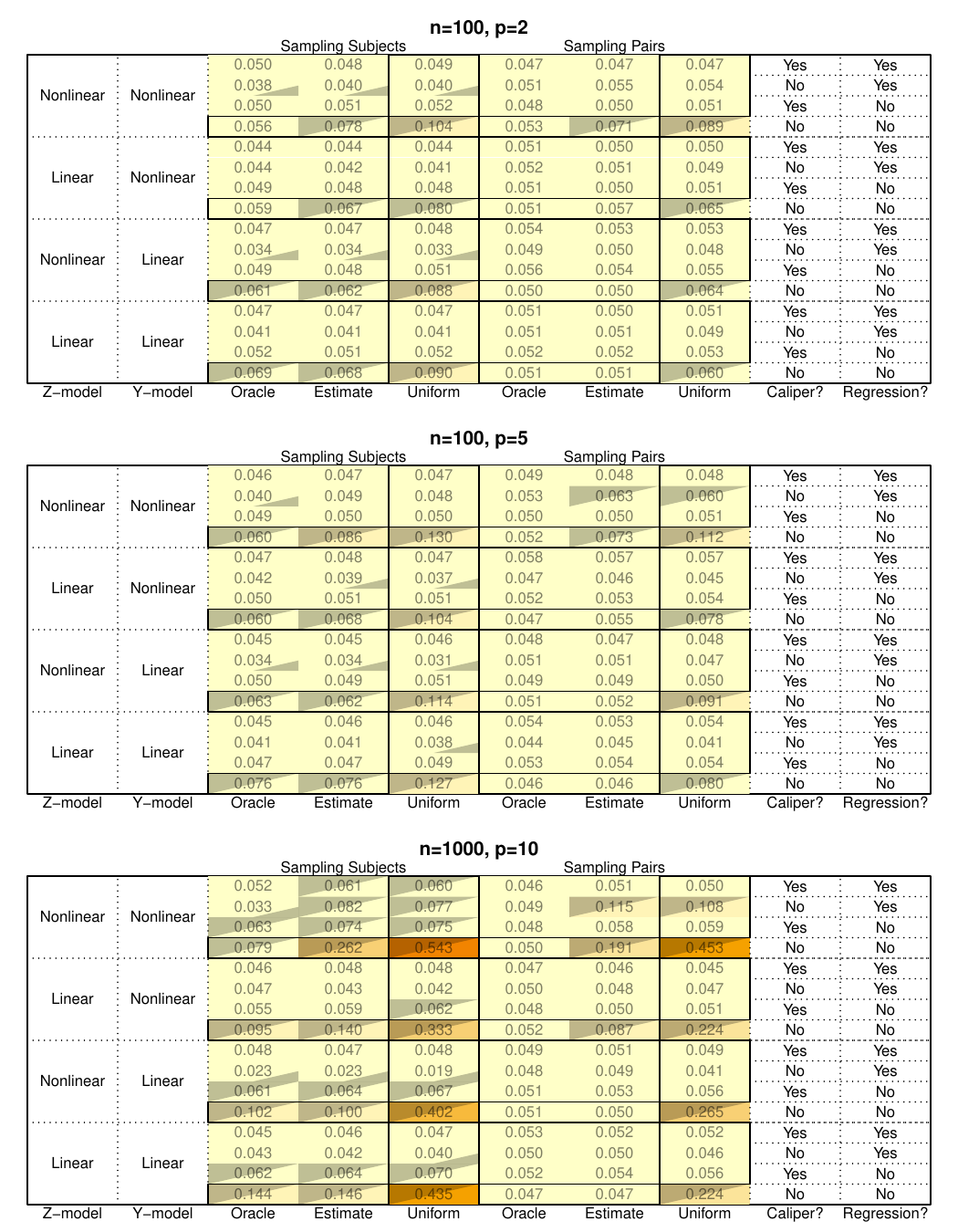}
\caption{\small Type I error results for uniform and covariate-adaptive inference across multiple simulation settings, using out-of-sample propensity scores.  For a more thorough description of how the tables are organized see the caption to Figure \ref{fig:ps_signal}.}
\label{fig:oos_heatmaps}
\end{figure}

\end{document}